\RequirePackage{fix-cm}
\documentclass[smallextended]{svjour3}       
\smartqed  
\usepackage{graphicx}
\usepackage{graphicx}
\usepackage{amsmath}
\usepackage{amsfonts}
\usepackage{epstopdf}
\usepackage{setspace}
\usepackage{graphicx}
\usepackage{subfigure}
\usepackage{subfig}
\usepackage[labelfont=bf]{caption}
\usepackage{algpseudocode}
\usepackage[ruled,linesnumbered,vlined]{algorithm2e}
\usepackage{verbatim}
\usepackage{bm}
\usepackage{multirow}
\usepackage{cite}
\usepackage{url}

\renewcommand{\algorithmicrequire}{\textbf{Input:}}  
\renewcommand{\algorithmicensure}{\textbf{Output:}} 
\newtheorem{defn}{\noindent $\mathbf{Definition}$}

\newtheorem{prop}[defn]{$\mathbf{Proposition}$}

\newtheorem{cor}[defn]{$\mathbf{Corollary}$}

\begin{document}

\title{The Theory of Computational Quasi-conformal Geometry on Point Clouds}


\author{Ting Wei Meng         \and
		Lok Ming Lui
}


\institute{T.W. Meng \at
              Department of Mathematics, The Chinese University of Hong Kong, Shatin, Hong Kong \\
              \email{twmeng@math.cuhk.edu.hk}           
           \and
           L.M. Lui (Corresponding author) \at
              Room 207, Lady Shaw Building, Department of Mathematics, The Chinese University of Hong Kong, Shatin, Hong Kong.
              \email{lmlui@math.cuhk.edu.hk}
}

\date{Received: date / Accepted: date}

\maketitle

\begin{abstract}
Quasi-conformal (QC) theory is an important topic in complex analysis, which studies geometric patterns of deformations between shapes. Recently, computational QC geometry has been developed and has made significant contributions to medical imaging, computer graphics and computer vision. Existing computational QC theories and algorithms have been built on triangulation structures. In practical situations, many 3D acquisition techniques often produce 3D point cloud (PC) data of the object, which does not contain connectivity information. It calls for a need to develop computational QC theories on PCs. In this paper, we introduce the concept of computational QC geometry on PCs. We define PC quasi-conformal (PCQC) maps and their associated PC Beltrami coefficients (PCBCs). The PCBC is analogous to the Beltrami differential in the continuous setting.  Theoretically, we show that the PCBC converges to its continuous counterpart as the density of the PC tends to zero. We also theoretically and numerically validate the ability of PCBCs to measure local geometric distortions of PC deformations. With these concepts, many existing QC based algorithms for geometry processing and shape analysis can be easily extended to PC data.



\end{abstract}

\section{Introduction}
\label{intro}
Quasi-conformal (QC) theory was firstly proposed by Ahlfors \cite{Ahlfors1} in 1953. Since then, it has become an important topic in complex analysis \cite{QC5}, \cite{QC4}, \cite{QC3}, \cite{QC2}. Applications have been found in various areas in mathematics and physics, including differential equations \cite{pde1}, \cite{pde2,pde3,pde4}, topology \cite{pde3}, complex dynamics \cite{complexdynamics}, physical simulation \cite{physics} and function theory \cite{pde3}.

Recently, computational QC geometry has been developed and different algorithms have been proposed to approximate QC maps on triangular meshes in a discrete setting. A discrete QC map is considered as an orientation preserving homeomorphism between meshes, which is piecewise linear on each triangular face. The Beltrami coefficient (BC) associated to a given discrete QC map can be computed. The discrete BC is a complex-valued (piecewise constant) function defined on each triangular face. According to the QC theories, the discrete BC measures the angular distortions of each triangular face under the QC map. Hence, the local geometric distortions under the discrete QC map can be captured by the discrete BC. Besides, given a discrete BC, its associated discrete QC map can be efficiently reconstructed by solving elliptic PDEs. Over the last few years, computational QC geometry has been successfully applied in medical imaging, computer graphics and computer vision, to solve important problems, such as surface registration \cite{LuiregInten}, \cite{LuiTMap}, \cite{Luihg}, \cite{LuiBHFHP}, \cite{LuiBHF}, \cite{LuiregNd}, \cite{Weiface}, shape analysis \cite{ShapeAnalysis2}, \cite{ShapeAnalysisNd}, \cite{ShapeAnalysisYamabe}, \cite{ShapeAnalysis1}, texture mapping \cite{LuiBeltramirepresentation}, video compression \cite{LuiBeltramirepresentation}, geometric modeling \cite{modeling} and others \cite{Daripa}, \cite{LuiCompression}, \cite{Mastin2}, \cite{inpaint}.

Computational QC theories and related algorithms have been built on triangulation structures. In practical situation, many 3D image acquisition techniques often produce point cloud (PC) data of the geometric object. However, it is still unclear whether existing QC theories can be extended to PCs. The challenges are two folded. Firstly, unlike a triangular mesh, there is no angle structure defined on a PC. Existing discrete QC theories are mainly related to the angular distortions under the deformation. It poses difficulties in defining
conformality of a PC deformation. Secondly, a general unstructured PC does not have connectivity information. Thus, the conventional definition of discrete QC maps as orientation preserving piecewise linear homeomorphisms is no longer valid for PCs. To the best of our knowledge, computational QC geometry on PCs has not been studied before. This motivates us to develop computational QC theories on PCs, so that existing QC based algorithms for geometry processing and shape analysis can be extended to PC data.

In this paper, we introduce the concept of computational QC geometry on PCs. We first give the definition of PC quasi-conformal (PCQC) maps and their corresponding PC Beltrami coefficients (PCBCs). The PCBC is analogous to the Beltrami differential in the continuous setting. Our main focus in this work is to study the relationship between PCQC maps and PCBCs. Theoretically, we show that the PCBC converges to its continuous counterpart as the density of PC tends to zero.
We also theoretically and numerically examine the ability of PCBCs to capture local geometric distortions of PC deformations.

The proposed theories of computational QC geometry on PCs provide us with a tool to study and control geometric patterns of deformations between PC shapes. With these concepts, many existing QC based algorithms for geometry processing and shape analysis, such as PC registration, data compression and shape recognition/classification, can be easily extended to PCs.

The rest of the paper is organized as follows. 
In Section 2, previous works closely related to this paper are reviewed. In Section 3, we give a brief introduction about conformal map and quasi-conformal map. In section 4, QC theory on point cloud is built. We define the PCQC map between two PCs approximating Riemann surfaces in $\mathbb{R}^2$ or $\mathbb{R}^3$. The PCBC associated to the PCQC map is then defined. The relationship between the PCBC and its continuous counterpart is also theoretically studied. We then explore the ability of PCBCs to capture local geometric distortions of PC deformations, including the change of angles within a neighborhood and the change of local covariance matrices. In Section 5, we show some experimental results on synthetic and real data to verify our propositions as described in Section 4.

\section{Related work}
For data with triangulation structures, different approaches have been proposed to compute QC maps from their associated BCs, which include the minimization of least-square Beltrami energy \cite{Zorin}, Quasi-Yamabe Flow \cite{LuiQuasiYamabe}, Beltrami Holomorphic Flow (BHF) \cite{LuiBHF}, discrete Beltrami Flow \cite{DBFHG}, Linear Beltrami Solver (LBS) \cite{LuiTMap}, QCMC \cite{LuiQCMC} and FLASH \cite{Luiflash,LuiflashD}. In this paper, we will extend some of the above ideas to PC data.

Although many works have been done to compute QC maps on mesh structures, computational QC theories on PCs have not yet been studied. Nevertheless, some works on PC parameterizations and registrations have been recently proposed. Below we list some of these works, which are closely related to ours.

Registration is an important process in various fields, such as computer vision and medical imaging. Its goal is to find a meaningful map between two corresponding domains. Several algorithms for the registration between PCs have been previously proposed. For an overview of this topic, we refer readers to the survey \cite{audette2000algorithmic}, \cite{tam2013registration}, \cite{van2011survey}.
These algorithms can mainly be divided into two categories. The first category involves solving some optimization models for all points of the PCs.
The most popular algorithm in this category is the Iterative Closest Point (ICP) method\cite{besl1992method}. Based on the ICP method, many other algorithms have been developed. For more details, we refer readers to \cite{pomerleau2013comparing,rusinkiewicz2001efficient}. The other category is called the feature-based registration model. The basic idea is to extract feature points, with which a registration map between PCs matching corresponding features can be obtained. The major tasks for these approaches include the extraction of feature points and the estimation of their correspondences, which are less related to our work.

Besides, the parameterization of PC data has also been widely studied. The main goal is to map a PC data onto a simple parameter domain, such as a compact 2D domain. The first PC parameterization algorithm has been developed by Floater and Reimers \cite{Floater1} in 2001. In that work, the authors proposed a "single patch" meshless parameterization algorithm by solving a linear system, which restricts each point as a linear combination of points in its neighborhood with some weight functions. Later, different weight functions have been developed \cite{Floater2,Floater3, Floater4, Floater5}. The algorithm has further been generalized to PC surfaces with a spherical topology \cite{Hormann1, Zwicker1}.
Besides, many mesh parameterization methods have been extended to PCs. These include: the Self Organizing Maps (SOM) approach \cite{Barhak1}, holomorphic 1-form method \cite{Guo1,Tewari1}, Multi-Dimensional Scaling (MDS) technique \cite{Miao1}, As-Rigid-As-Possible (ARAP) method \cite{Zhang1}, and Periodic Global Parameterization (PGP) \cite{Li1,Li2}. In addition, Zwicker et al. \cite{Zwicker2} proposed to obtain parameterization through energy minimization. Wang et al. \cite{Wang1} also developed a simple parameterization algorithm by projecting and unfolding. More recently, Liang et al. \cite{Hongkai1,Hongkai2} proposed to obtain spherical parameterizations of PCs using the Moving Least Square method.
To obtain a meaningful PC parameterization, accurate measures of geometric distortions under the parameterization are necessary. In this work, our goal is to develop computational QC theories on PCs, which study local geometric distortions under PC deformations.

\section{Mathematical background}
In this section, we introduce some basic concepts about conformal and QC maps. These two kinds of maps have been widely used in geometry processing and computer vision. For more details, we refer readers to \cite{Gardiner}.

A conformal map is a diffeomorphism between two Riemann surfaces that satisfies Cauchy-Riemann equation. It preserves the local angle structure and maps infinitesimal circle to infinitesimal circle. Riemann mapping theorem states that any simply-connected compact open Riemann surface can be conformally mapped to a unit disk. Furthermore, this conformal parameterization is unique up to a Mobi\"us transformation.

A generalization of the conformal map is called the quasi-conformal (QC) map. An orientation-preserving diffeomorphism $f:\mathbb{C}\to \mathbb{C}$ is quasi-conformal if $\|\mu(f)\|_\infty$ $ < 1$, where $\mu(f)$ is called the Beltrami coefficient(BC) of $f=u+iv$ defined by
\begin{equation}
\mu(f) = \frac{\partial_{\bar{z}}f}{\partial_z f} =
\frac{(u_x-v_y) + i(v_x + u_y)}{(u_x+v_y)+i(v_x-u_y)}
\end{equation}
Given a feasible BC function $\mu$, a unique QC map $f$ can be computed by solving the Beltrami equation $\partial_{\bar{z}} f = \mu \partial_z f$. A QC map is conformal if and only if its BC is zero at any point. Hence, a conformal map is a special case of QC maps.

Given a feasible BC, its corresponding QC map always exists and is unique if $0,1,\infty$ are fixed. Therefore, the set of QC maps and the set of feasible BCs has one-to-one relationship up to normalization.

Moreover, BC itself can captures the geometric information of a QC map $f$. Near any point $p$, a QC map $f$ can be locally linearized as $f(z)-f(p) = \partial_z f|_p(dz + \mu(p) d\bar{z})$, so that $f$ transforms circles to ellipses in infinitesimal sense. The dilation is given by $\frac{1+|\mu(p)|}{1-|\mu(p)|}$. The stretch direction is also controlled by $\mu(p)$, as illustrated in Figure \ref{fig:QCstretch}.

\begin{figure}[h!]
\centering
\includegraphics[height=1.35in]{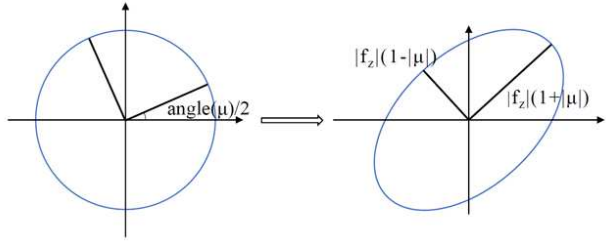}
\caption{Quasi-conformal deformation in the infinitesimal sense\label{fig:QCstretch}}
\end{figure}

When two maps $f$ and $g$ are given, we have a formula for the BC of the composition map $f\circ g^{-1}$:
\begin{equation}
\mu(f\circ g^{-1})\circ g = \frac{\mu(f)-\mu(g)}{1 - \mu(f)\overline{\mu(g)}}\cdot \frac{g_z}{\overline{g_z}}
\end{equation}

From the above formula, it is easy to observe that BC is not affected by left composition of a conformal map. However, BC is rotated under right composition of a conformal map. More precisely, if $f$ is conformal, then $\mu(f\circ g) = \mu(g)$, and $\mu(g\circ f) = \mu(g)\overline{f_z}/f_z$.

The definition of QC maps can also be generalized to Riemann surfaces. In this case, Beltrami differentials have to be used, instead of Beltrami coefficients. A Beltrami differential $\mu(z) \frac{\overline{dz}}{dz}$ is defined on $S_1$ by assigning complex-valued $L_\infty$ function $\mu_\alpha$ to each chart $(U_{\alpha},\phi_{\alpha})$ such that
\begin{equation*}
\mu_{\alpha}(z_{\alpha})\frac{\overline{dz_{\alpha}}}{dz_{\alpha}} = \mu_{\beta}(z_{\beta})\frac{\overline{dz_{\beta}}}{dz_{\beta}},
\end{equation*}
on $U_\alpha \cap U_\beta$ where $(U_{\beta},\phi_{\beta})$ is another arbitrary chart of $S_1$. A map $f:S_1\to S_2$ between two Riemann surfaces $S_1$ and $S_2$ is said to be a QC map associated to $\mu(z) \frac{\overline{dz}}{dz}$, if $f$ restricted to each chart $U_{\alpha}$ is a QC map with BC $\mu_\alpha$. Therefore, QC theories on 2D complex plane can be naturally generalized to Riemann surfaces.

\section{Model Setting}
In this section, we develop the concept of computational QC geometry on PC structures. We define PC quasi-conformal (PCQC) maps and its associated PC Beltrami coefficients (PCBCs), which capture local geometric distortions under PCQC maps. We first build QC theories on 2D PCs. The extension of the developed theories to 3D PCs sampled from Riemann surfaces embedded in $\mathbb{R}^3$ is then discussed.

Inspired by continuous QC theories, our definitions of QC geometry on PCs rely on the approximation of partial derivatives. The approximation of partial derivatives on PCs have been widely studied and various methods have been proposed recently. The most common approach is done by minimizing the least square error between sampled values and a linear combination of some base functions. In particular, this method is called the Moving Least Square method (MLS), if the base functions are chosen to be a set of polynomials. MLS will be used in this paper.
Before the discussion of our proposed QC model, we give a brief introduction about the general setting of PCs and the MLS method. For more details, we refer the readers to \cite{Wendland1,Mirzaei1,Mirzaei2}.

For any PC $\mathcal{P}$, we can write it as an ordered set $\mathcal{P} = \{p_1,p_2,...,p_N\}\subset \mathbb{R}^2$. $\mathcal{P}$ can also be represented by its matrix form $P=[p_1,p_2,...,p_N]^T\in M_{N\times 3}(\mathbb{R})$. Similarly, a PC function $f:\mathcal{P}\to \mathbb{R}$ can be identified with a matrix $F=[f(p_1),f(p_2),...,f(p_N)]^T\in \mathbb{R}^N$. In this paper, we always use the uppercase letter for the matrix form of a PC function represented by the corresponding lowercase letter, when it is not specified.

Suppose $\mathcal{P}$ is sampled from a domain $\mathcal{D}$ in $\mathbb{R}^2$ and $\bold{f}: \mathcal{P}\to \mathbb{R}$ is a function on $\mathcal{P}$. For each point $x\in \mathcal{D}$, one can calculate a polynomial $q_x$ to approximate $f$ near $x$ by using MLS. To simplify calculation process, we just consider quadratic polynomial approximation.
Let $w: [0,+\infty) \to \mathbb{R}$ be a weight function, compactly supported in $[0, 1)$, positive on $[0,1/2)$, and with even extension in $C^3(\mathbb{R})$. Define maps $q,q_1,q_2:\mathcal{D}\to \mathbb{R}^6$ by $q(x_1,x_2)=[1,x_1,x_2,x_1^2,x_1x_2,x_2^2]^T$, $q_1(x_1,x_2)=\partial_1 q= [0,1,0,2x_1,x_2,0]^T$, $q_2(x_1,x_2)=\partial_2 q = [0,0,1,0,x_1,2x_2]^T$. Let $\delta = C_\delta h_{\mathcal{P},\mathcal{D}}$ be the neighborhood radius parameter which controls the size of influenced neighborhood, where
$$
C_\delta = \frac{128(1+\sin\theta)^2}{3\sin^2\theta}
$$

Then MLS computes a local approximation $q_x(y) = c_x^Tq(y)$ by minimizing the weighted least square error
$$
c_x = argmin_c \, \sum_{i=1}^N w\left(\frac{\|x-p_i\|}{\delta}\right)\left(c^Tq(p_i)-f_i\right)^2
$$
The solution $c_x$ has a closed formula, $c_x = (Q^TW(x)Q)^{-1}Q^TW(x)F$, where $Q$ and $F$ are the matrix forms of $q$ and $f$ respectively, $W(x)$ is a diagonal matrix whose diagonal element $W_{ii} = w(\|x-p_i\|/\delta)$. Here, nonsigularity of matrix $Q^TW(x)Q$ is required, which will be assumed for all PCs in this paper.

Denote $A_x = (Q^TW(x)Q)^{-1}Q^TW(x)$, then $q_{x}(y) = q(y)^TA_xF$, and one can calculate partial derivatives of $q_{x}$, $\partial_jq_{x}(y) = q_j(y)^TA_xF$, which are called diffuse derivatives.

Globally, MLS gives a function $\hat{f}(x) = q(x)^TA_xF$ defined on $\mathcal{D}$, which approximates the map $f$. Define $\Phi_i(x) = (q(x)^TA_x)_i$, then $\hat{f} = \sum_{i=1}^N f_i\Phi_i(x)$ is a linear combination of $\Phi_i$. Here, the shape functions $\Phi_i$ only depend on PC data. Then one can compute partial derivatives of $\hat{f}$ and get $\partial_j \hat{f} = \sum_{i=1}^N f_i\partial_j\Phi_i(x)$, $j=1,2$, which are called standard derivatives.

In order to analyze error of MLS approximation, there are some requirements for the shape of domain and also the distribution of PCs, which are given in the following definition \cite{Wendland1}.

\begin{defn} \label{Def:feasible pc}
A domain $\mathcal{D}\subseteq \mathbb{R}^2$ is said to satisfy an interior cone condition with parameter $r>0$ and $\theta\in (0,\pi/2)$ if for each $x\in \mathcal{D}$ there exists a unit vector $d(x)$ such that $C(x,d,\theta, r)\subseteq \mathcal{D}$, where
$$
C(x,d,\theta, r) = \{x+ty: y\in \mathbb{S}^1, \,y^Td(x)\geq \cos \theta, \,0\leq t\leq r\}
$$
Let $\mathcal{P}$ be a point cloud sampled from a simply connected compact open 2-manifold $D$ where $D$ can be either a domain in $\mathbb{R}^2$ or a Riemann surface in $\mathbb{R}^3$. Fill distance $h_{\mathcal{P},D}$ and separation distance $q_{\mathcal{P}}$ are defined by
\begin{align*}
&h_{\mathcal{P},D} = \sup_{x\in D} \min_{p\in \mathcal{P}} \|x-p\|\\
&q_{\mathcal{P}} = \frac{1}{2}\min_{p_1,p_2\in \mathcal{P}\atop p_1\neq p_2} \|p_1-p_2\|
\end{align*}
$\mathcal{P}$ is said to be quasi-uniform with positive constant $c_{qu}$ if $q_{\mathcal{P}} \leq h_{\mathcal{P},D} \leq c_{qu}q_{\mathcal{P}}$.
\end{defn}

There are some error analysis for MLS, and we will use those described in \cite{Mirzaei2,Mirzaei1} here. In their paper, they defined a semi-norm $\|\cdot\|_{C^k(\mathcal{D})}$ for any real function $g\in C^m(\mathcal{D};\mathbb{R})$ and any complex function $f=u+iv$ where $u,v \in C^m(\mathcal{D};\mathbb{R})$, and $k\leq m$, as follows.
$$
\|g\|_{C^k(\mathcal{D})} = \max_{|\alpha|=k} \|D^\alpha g\|_{L_\infty(\mathcal{D})}\ \ \mathrm{and\ \ }
\|f\|_{C^k(\mathcal{D})} = \max\{\|u\|_{C^k(\mathcal{D})},\|v\|_{C^k(\mathcal{D})}\}
$$
For a simple case when $k=1$, $\|g\|_{C^1(\mathcal{D})} = \max\{\|\partial_1 g\|_{L_\infty(\mathcal{D})}, \|\partial_2 g\|_{L_\infty(\mathcal{D})}\}$.

The error of MLS approximation is bounded using semi-norm of $f$ and fill distance of the PC.
When domain $\mathcal{D}$ and quasi-uniform constant $c_{qu}$ are fixed, there exists one constant $C_{MLS}$ such that for arbitrary map $f\in C^3(\mathcal{D})$, and any PC $\mathcal{P}$ sampled from $\mathcal{D}$ satisfying properties in Definition \ref{Def:feasible pc} with fill distance $h<h_0 := 2r/C_\delta$,
\begin{align*}
&|f(x) - \hat{f}(x)| \leq C_{MLS} \|f\|_{C^3(\mathcal{D}^*)}h^3\\
&|\partial_i f(x) - \partial_i \hat{f}(x)| \leq C_{MLS} \|f\|_{C^3(\mathcal{D}^*)}h^2\\
&|\partial_i f(x) - \partial_i q_{x}(x)| \leq C_{MLS} \|f\|_{C^3(\mathcal{D}^*)}h^2
\end{align*}
where $\mathcal{D}^*$ is the closure of $\cup_{x\in \mathcal{D}} B(x,2r)$.

\subsection{Quasi-conformal maps between point clouds on complex domains}
In this part, we only consider the maps defined on an arbitrary compact domain $\mathcal{D}\subseteq \mathbb{R}^2$ satisfying interior cone condition. Basic definitions of PC quasi-conformal(PCQC) maps and PC Beltrami coefficients(PCBCs), as well as their relationships, will be discussed.

\begin{defn}
A map $\bold{f}$ is called a point cloud map if it is defined on a point cloud $\mathcal{P}$ by assigning each point $p$ a vector $f(p)$ in either $\mathbb{R}^2$ or $\mathbb{R}^3$.
When $\mathcal{P}$ and $\bold{f}(\mathcal{P})$ are both in $\mathbb{R}^2$, $\bold{f}$ is called
a parameterization map of $\mathcal{P}$ if it is a injective map and $\bold{f}(\mathcal{P})\subseteq \mathcal{D}$.
\end{defn}

Similar to continuous QC theory, diffuse or standard PCBC is defined according to the Beltrami's equation, using the diffuse or standard derivative approximation respectively. Explicit formulas are given below.

\begin{defn} \label{Def:mu pc}
Given a point cloud $\mathcal{P} =\{p_i\in \mathbb{R}^2: i=1,...,N\}$ and a target point cloud function $\bold{f}:\mathcal{P} \to \mathbb{R}^2$, where $\bold{f}=(\bold{u},\bold{v})^T$, diffuse Beltrami coefficient $\tilde{\mu}:\mathcal{D}\to \mathbb{C}$ is defined by
$$
\tilde{\mu}(x) = \frac{
\begin{bmatrix}
q_1(x)^T & q_2(x)^T
\end{bmatrix}
\begin{bmatrix}
A_x & iA_x\\
iA_x & -A_x
\end{bmatrix}
\begin{bmatrix}
U\\V
\end{bmatrix}
}{\begin{bmatrix}
q_1(x)^T & q_2(x)^T
\end{bmatrix}
\begin{bmatrix}
A_x & iA_x\\
-iA_x & A_x
\end{bmatrix}
\begin{bmatrix}
U\\V
\end{bmatrix}
}
$$
Standard Beltrami coefficient $\hat{\mu}:\mathcal{D}\to \mathbb{C}$ is defined by
$$
\hat{\mu}(x) = \frac{
\begin{bmatrix}
q_1(x)^T & q_2(x)^T
\end{bmatrix}
\begin{bmatrix}
A_x & iA_x\\
iA_x & -A_x
\end{bmatrix}
\begin{bmatrix}
U\\V
\end{bmatrix}
+
\begin{bmatrix}
q(x)^T & q(x)^T
\end{bmatrix}
\begin{bmatrix}
\partial_1 A_x & i\partial_1 A_x\\
i\partial_2 A_x & -\partial_2 A_x
\end{bmatrix}
\begin{bmatrix}
U\\V
\end{bmatrix}
}{\begin{bmatrix}
q_1(x)^T & q_2(x)^T
\end{bmatrix}
\begin{bmatrix}
A_x & iA_x\\
-iA_x & A_x
\end{bmatrix}
\begin{bmatrix}
U\\V
\end{bmatrix}
+
\begin{bmatrix}
q(x)^T & q(x)^T
\end{bmatrix}
\begin{bmatrix}
\partial_1 A_x & i\partial_1 A_x\\
-i\partial_2 A_x & \partial_2 A_x
\end{bmatrix}
\begin{bmatrix}
U\\V
\end{bmatrix}
}
$$
where for $j=1,2$,
\begin{align*}
\partial_j A_x = &-(P^TW(x)P)^{-1}P^T(\partial_jW(x))P(P^TW(x)P)^{-1}P^TW(x)\\
&+ (P^TW(x)P)^{-1}P^T\partial_jW(x)
\end{align*}

Discrete diffuse Beltrami coefficient \bm{$\tilde{\mu}$} is defined by \bm{$\tilde{\mu}$} $ = [\tilde{\mu}(p_1),\tilde{\mu}(p_2), ..., $ $\tilde{\mu}(p_N)]^T$.
Similarly, discrete standard Beltrami coefficient \bm{$\hat{\mu}$} is defined by \bm{$\hat{\mu}$} $ = [\hat{\mu}(p_1), \hat{\mu}(p_2),$ $...,\hat{\mu}(p_N)]^T$.
\end{defn}

We remark that, in our practical implementation, the discrete diffuse PCBC is often used because of its simplicity to calculate. On the other hand, the standard PCBC is exactly the BC of the global approximation function $\hat{f}$ in the continuous setting.

With the definition of PCBCs, PCQC maps can be easily defined as follows.

\begin{defn}
A point cloud map $\bold{f}$ is called diffuse point cloud quasi-conformal map if \bm{$\tilde{\mu}}$ is well-defined and $|\bm{\tilde{\mu}}|<1$ on each point in $\mathcal{P}$.

$\bold{f}$ is called standard point cloud quasi-conformal map if \bm{$\hat{\mu}$} is well-defined and $|\bm{\hat{\mu}}|<1$ on each point in $\mathcal{P}$.
\end{defn}

As we will discuss later, the PCBC measures the local geometric distortion under the PCQC map, which is analogous to the continuous Beltrami coefficient. In fact, the PCQC has a close relationship with its continuous counterpart, up to some error controlled by the fill distance of the PC data.

In order to analyze the error, we assume that $\mathcal{D}$ is a simply connected compact domain in $\mathbb{R}^2$ satisfying interior cone condition in Definition \ref{Def:feasible pc}, and $f:\mathcal{D}\to \mathbb{R}^2$ is a QC map in $C^3(\mathcal{D})$ with BC $\mu(f)$. Any PC $\mathcal{P}$ sampled from $\mathcal{D}$ is assumed to be a feasible PC, i.e. it is quasi-uniform with fixed constant $c_{qu}>0$, and its fill distance $h \leq h_0$.

\begin{prop} \label{Prop:muerr}
Let $\mathcal{P}$ be a point cloud sampled from $\mathcal{D}$ with fill distance $h$, and $\bold{f}$ be the point cloud map corresponding to $f$. Then there exist some constants $C_1(f)$ and $C_2(f)$, such that if $h \leq C_1(f)$, both diffuse Beltrami coefficient $\tilde{\mu}$ and standard Beltrami coefficient $\hat{\mu}$ have error bound
\begin{align*}
& |\mu(x) - \tilde{\mu}(x)| \leq C_2(f)h^2\\
& |\mu(x) - \hat{\mu}(x)| \leq C_2(f)h^2
\end{align*}
\end{prop}

\begin{proof}
Let $g:\mathcal{D}\to \mathbb{R}$ be an arbitrary function. We use the notation $\partial_j q_{g,p}(x)$ to denote diffuse derivatives approximating partial derivatives of $g$ near an arbitrary point $p$. Let $\tilde{E}_{g,j}(x) = \partial_j q_{g,x}(x) - \partial_j g(x)$ be the error of diffuse derivatives.

As stated above, for each point $x\in \mathcal{D}$,
\begin{align*}
&|\tilde{E}_{g,j}(x)| \leq C_{MLS}  \|g\|_{C^3(\mathcal{D}^*)}h^2\\
\end{align*}
From Definition \ref{Def:mu pc}, the following equation can be obtained.
\begin{align*}
\tilde{\mu}(x) &= \frac{(\partial_1 q_{u,x} - \partial_2 q_{v,x}) + i(\partial_1 q_{v,x} + \partial_2 q_{u,x})}{(\partial_1 q_{u,x} + \partial_2 q_{v,x}) + i(\partial_1 q_{v,x} - \partial_2 q_{u,x})}\\
\end{align*}
Then by comparing $\tilde{\mu}$ with $\mu(f)$, one can derive
\begin{align*}
\tilde{\mu}
&= \frac{(\partial_1 u - \partial_2 v) + i(\partial_1 v + \partial_2 u) + e_1}{(\partial_1 u + \partial_2 v) + i(\partial_1 v - \partial_2 u) + e_2}
\end{align*}
where $|e_j|\leq 2\sqrt{2} C_{MLS} \|f\|_{C^3(\mathcal{D}^*)} h^2$, $j=1,2$.

According to the definition of Beltrami coefficient, $\partial_z f \neq 0$ everywhere in the compact domain $\mathcal{D}$, which implies that $|\partial_z f|$ has a positive lower bound, denoted by $L$. Without loss of generality, assume $\|f\|_{C^3(\mathcal{D}^*)}>0$. Let $C_1(f) = \sqrt{\frac{L}{4\sqrt{2}C_{MLS}\|f\|_{C^3(\mathcal{D}^*)}}}$, then $C_1(f)>0$, and $h\leq C_1(f)$ implies $|e_2|\leq L/2$.
Then,
$$
|\tilde{\mu} - \mu(f)| \leq \frac{|e_1||\partial_z f| + |e_2||\partial_{\bar{z}} f|}{L^2/2} \leq C_2(f)h^2
$$
where $C_2(f) =  \frac{32C_{MLS}\|f\|_{C^3(\mathcal{D}^*)} \|f\|_{C^1(\mathcal{D})}}{L^2}$.

With similar argument, one can prove similar result for standard Beltrami coefficient.
\qed
\end{proof}

\bigbreak
From now on, we assume the fill distance of PC satisfies the condition $h\leq C_1(f)$. From the above proposition, we see that both diffuse and standard PCBCs give good approximations of the continuous BC $\mu(f)$.

Next, we will study theoretically the ability of PCBCs to capture local geometric distortions under PCQC maps. In fact, PCBCs capture local changes of many statistical and geometric properties on PCs under the PCQC maps, such as the local angle structure and covariance matrices.

\begin{prop}\label{Prop:angle}
Let $\{\mathcal{P}_n\}$ be a sequence of point clouds sampled from $\mathcal{D}$ with fill distance $h_n$ goes to $0$, and $\bold{f}_n:\mathcal{P}_n\to \mathbb{R}^2$ be the point cloud map corresponding to $f$. If sup-norm of discrete standard or diffuse Beltrami coefficient converges to $0$, then
$$
\max_{p_0\in \mathcal{P}_n \atop p_1,p_2\in S_n(p_0)\backslash \{p_0\}} \left|\alpha(\bold{f}_n(p_1),\bold{f}_n(p_0),\bold{f}_n(p_2)) - \alpha(p_1,p_0,p_2) \right| \leq Ch_n
$$
where $S_n(p_0) = B_{\delta_n}(p_0)\cap \mathcal{P}_n$, and $\alpha(p_1,p_0,p_2)$ denotes the angle between $p_1-p_0$ and $p_2-p_0$.
\end{prop}

\begin{proof}

Let $\mu$ be Beltrami coefficient of $f$, and $\bm{\tilde{\mu}}_n$ be discrete diffuse Beltrami coefficient of $\bold{f}_n$. Fix integer $n$, let $x$ be an arbitrary point in domain, $p\in \mathcal{P}_n$ be the nearest point to $x$, then we have $\|x-p\|\leq h_n$. According to Taylor's theorem, there exists a point $\xi \in \mathcal{D}$ such that
$$
|\mu(x)-\mu(p)|=|\nabla \mu(\xi)(x-p)|\leq 2\|\mu\|_{C^1(\mathcal{D})}\|x-p\|\leq 2\|\mu\|_{C^1(\mathcal{D})}h_n
$$
From Property \ref{Prop:muerr}, $|\bm{\tilde{\mu}}_n(p) - \mu(p)| \leq C_0h_n^2$ for some constant $C_0$.
Then $|\mu(x)| \leq |\bm{\tilde{\mu}}_n(p)|+2\|\mu\|_{C^1(\mathcal{D})} h_n + C_0 h_n^2$, which goes to zero when $n$ goes to infinity. Hence $f$ is a conformal map.

Fix integer $n$. Let $\alpha_1 = \alpha(p_1,p_0,p_2)$, $\alpha_2 = \alpha(\bold{f}_n(p_1),\bold{f}_n(p_0),\bold{f}_n(p_2))$, where $p_1, p_2\in S_n(p_0)\backslash \{p_0\}$. Then, $\bold{f}_n(p_j) - \bold{f}_n(p_0) = \nabla f(p_0)(p_j-p_0)+e_j$, where $\|e_j\|\leq 2\sqrt{2}\|f\|_{C^2(\mathcal{D})}\|p_j-p_0\|^2$, $j=1,2$. Therefore,
\begin{align*}
|\alpha_1 - \alpha_2|
&\leq \arcsin\left(\frac{\|e_1\|}{\|\nabla f(p_0)(p_1-p_0)\|}\right) + \arcsin\left(\frac{\|e_2\|}{\|\nabla f(p_0)(p_2-p_0)\|}\right)\\
&\leq \frac{\pi}{2}\cdot \sqrt{\frac{1+|\mu(p_0)|}{\det(\nabla f(p_0))(1-|\mu(p_0)|)}}\left(\frac{\|e_1\|}{\|p_1-p_0\|} + \frac{\|e_2\|}{\|p_2-p_0\|}\right)\\
&\leq \pi \|f\|_{C^2(\mathcal{D})} \sqrt{\frac{2(1+\|\mu\|_\infty)}{L_1(1-\|\mu\|_\infty)}}(\|p_1-p_0\|+\|p_2-p_0\|) \leq C_1h_n
\end{align*}
where $L_1 = \min \det(\nabla f)>0$, and
$C_1 = 2\sqrt{2}\pi C_\delta\|f\|_{C^2(\mathcal{D})}\sqrt{\frac{1+\|\mu\|_\infty}{L_1(1-\|\mu\|_\infty)}}$ is a constant independent of $p_0$, $p_1$ and $p_2$. Therefore,
$$
\max_{p_0\in \mathcal{P}_n \atop p_1,p_2\in S_n(p_0)\backslash \{p_0\}} \left|\alpha(\bold{f}_n(p_1),\bold{f}_n(p_0),\bold{f}_n(p_2)) - \alpha(p_1,p_0,p_2) \right| \leq C_1 h_n
$$
Similarly, one can prove the statement for standard Beltrami coefficient case.
\qed
\end{proof}

The above proposition is analogous to the fact that angles are preserved under a conformal map. When the PCBC is small enough, the angle structure is preserved under the PCQC deformation, up to a tolerable error controlled by the fill distance.

On the other hand, statistical approaches are often used to analyze PC structures. For instance, the correlation and directional information of a PC structure are crucial, which can be used in many famous algorithms for determining features and shape analysis. The PCBC captures information about local changes in statistical properties of a PC structure under a PCQC map.

\begin{prop} \label{Prop:Pca} 
Let $f$ be a map defined on $\mathcal{D}$, and $\mathcal{P}$ be a point cloud sampled from $\mathcal{D}$ with fill distance $h$. Let $p\in \mathcal{P}$ such that $B_\delta(p)\subseteq \mathcal{D}$, and $S = B_\delta(p) \cap \mathcal{P}$. Let covariance matrix $M_1$ for $S$ has eigenvalues $\lambda_1, \lambda_2$, and covariance matrix $M_2$ for $f(S)$ has eigenvalues $\lambda_3, \lambda_4$, where $\lambda_1 \geq \lambda_2>0$, and $\lambda_3\geq \lambda_4>0$. Denote $v_2 = [v_{21},v_{22}]^T$ to be the unit eigenvector of $M_1$ with respect to $\lambda_2$. Let $\tilde{\mu}(p)$ be discrete diffuse Beltrami coefficient defined on $p$.
\\
(a). For $h$ small enough, there exists some constant $C$ depending on $f$ and independent of $p$ such that
$$
\left|\frac{\lambda_3}{\lambda_4} - \left(\frac{1+|T|}{1-|T|}\right)^2\right| \leq Ch
$$
where
$$
T = \frac{(\sqrt{\lambda_1}+\sqrt{\lambda_2})\tilde{\mu}(p) -(\sqrt{\lambda_1}-\sqrt{\lambda_2})(v_{21}+iv_{22})^2}{\sqrt{\lambda_1}+\sqrt{\lambda_2} -(\sqrt{\lambda_1}-\sqrt{\lambda_2})(v_{21}-iv_{22})^2\tilde{\mu}(p)}
$$
\\
(b). Assume
$\zeta = \mu(p) - \frac{\sqrt{\lambda_1}-\sqrt{\lambda_2}}{\sqrt{\lambda_1}+\sqrt{\lambda_2}}(v_{21}+iv_{22})^2
$, and $h=o(|\zeta|)$.
Denote $\theta = angle(T)/2$, and $u_0 = [\cos(\theta), \sin(\theta)]^T$.
Let vector $w_0$ be the unit eigenvector of $M_2$ with respect to $\lambda_3$, such that $w_0\cdot (\nabla f(p) (M_1 + \sqrt{\lambda_1\lambda_2}I) u_0) \geq 0$. Then when $h$ is small enough, there exist constants $C_1$, $C_2$ depending on $f$ and independent of $p$ such that
$$
\left|w_0 - \frac{\nabla f(p)(M_1 + \sqrt{\lambda_1\lambda_2}I) u_0}{\|\nabla f(p)(M_1 + \sqrt{\lambda_1\lambda_2}I) u_0\|}\right| \leq C_1h + C_2\frac{h}{|\zeta|}
$$
\end{prop}

The proof of this proposition is long and is attached in the appendix. This result also holds for standard PCBCs. In particular, when PC $\mathcal{P}$ is regular, we have the following corollary.

\begin{cor} \label{Prop:PcaId}
With the above conditions, further assume that $\lambda_1 = \lambda_2$.
\\
(a). For $h$ small enough, there exists some constant $C$ independent of $\mathcal{P}$ such that
$$
\left|\frac{\lambda_3}{\lambda_4} - \left(\frac{1+|\tilde{\mu}(p)|}{1-|\tilde{\mu}(p)|}\right)^2\right| \leq Ch
$$
\\
(b). Assume $\mu(p)\neq 0$. Denote $u_0 = [\cos(\theta), \sin(\theta)]^T$, where $\theta = angle(\tilde{\mu}(p))/2$. Let $w_0$ be the unit eigenvector of $M_2$ with respect to $\lambda_3$, such that $w_0\cdot (\nabla f(p)u_0)$ $\geq 0$. Then when $h$ is small enough, there exists constant $C$ depending on $f$ and independent of $p$ such that
$$
\left|w_0 - \frac{\nabla f(p) u_0}{\|\nabla f(p) u_0\|}\right| \leq Ch
$$
\end{cor}

\bigbreak

Corollary 8a is analogous to the fact that the dilation of the infinitesimal ellipse from a deformed infinitesimal circle under the QC map can be determined from the BC (denoted by $\mu$). More specifically, the dilation $K$ is given by: $K = (1+|\mu|)/(1-|\mu|)$.

Furthermore, according to QC theories, $\nabla f(p)u_0$ is approximately parallel to $u_1$, where $angle(u_1)=angle(-\tilde{\mu}(f^{-1})\circ f(p))/2$ and $\tilde{\mu}(f^{-1})$ is the standard or diffuse PCBC of $f^{-1}|_{f(\mathcal{P})}$. Therefore, the eigenvectors of $M_2$ in Corollary \ref{Prop:PcaId} can be approximated by PCBCs of $f^{-1}$ and $f$. By a similar argument, one can also use $M_1$ and PCBCs of $f$ and $f^{-1}$ to approximate the eigenvectors of $M_2$ in Proposition \ref{Prop:Pca}, without using $\nabla f(p)$.

The above propositions tell us PCBCs can be used to control local geometric distortions under the PCQC maps. For example, an optimal PCQC map that minimizes the local geometric distortions for PC registration can be obtained by minimizing its PCBC. Shape analysis of PC structures can also be done using PCBCs.

\subsection{Quasi-conformal maps between point cloud surfaces}

In the last subsection, we develop PCQC theories on $\mathbb{R}^2$. The theories can be extended to point cloud surfaces, that is, PC data sampled from 2D Riemann surfaces embedded in $\mathbb{R}^3$. Given a map between two PC surfaces, we will define its PC Beltrami representation(PCBR). The PCBR measures local geometric changes under its associated PCQC map.

To begin with, we have to impose some requirements on the PC surface. A Riemann surface can be linearized near any point in infinitesimal sense. Analogous to that fact, we require that a small neighborhood near any point of the PC can always be injectively projected to a plane.

\begin{defn} \label{Def:pcsurf}
Let $\mathcal{P}$ be a point cloud sampled from a simply connected Riemann surface $\mathcal{S}$ in $\mathbb{R}^3$ with fill distance $h$. For each point $p\in \mathcal{P}$, its $d_p-$ neighborhood is defined as $\mathcal{N}_{p,d_p} = B_{d_p}(p)\cap \mathcal{P}$. A point cloud $\mathcal{P}$ is called $d-$point cloud surface if for each point $p\in \mathcal{P}$, there exists a unit vector $v_p$ such that $|(x-p,v_p)|<q_{\mathcal{P}}$ for all $x\in \mathcal{N}_{p,d(p)}$ where $d$ is a positive point cloud function.
\end{defn}

\bigbreak

To define PCBR for a map between two PC surfaces, we use the similar idea as in the continuous case. According to QC theories, the Beltrami differential of a map between two Riemann surfaces is defined based on the projected map between the coordinate charts of the surfaces. In other words, the two Riemann surfaces are conformally parameterized onto simply-connected patches in $\mathbb{R}^2$. The Beltrami differential is defined by the Beltrami coefficient of the projected map between the conformal parameter domains. Thus, in order to define PCBR, we need to give a definition of PC conformal parameterization.

We first define PCBCs from a planar PC to a PC surface, which will then be used to define PC conformal parameterization.
Let $f=(u,v,w): \mathcal{D}\to \mathcal{S}$ be a QC map from a simply connected compact domain $\mathcal{D}\subseteq \mathbb{R}^2$ to a Riemann surface $\mathcal{S}\subseteq \mathbb{R}^3$, then it can be locally approximated near any point $x$ by function $q_{x}(y)=q(y)^TA_xF$, where $F=[U,V,W]$. From this approximation, we can define the PCBCs of function $f$ as follows.

\begin{defn}\label{Def:3dmu}
Given a point cloud $\mathcal{P}\subseteq \mathcal{D}$, a target point cloud function $\bold{f}=(\bold{u},\bold{v},\bold{w})^T: \mathcal{P} \to \mathbb{R}^3$, diffuse Beltami coefficient $\tilde{\mu}:\mathcal{D}\to \mathbb{C}$ is defined by
$$
\tilde{\mu}(x) = \frac{
\begin{bmatrix}
q_1(x)^T & q_2(x)^T
\end{bmatrix}
\begin{bmatrix}
B & iB\\
iB & -B
\end{bmatrix}
\begin{bmatrix}
q_1(x)\\q_2(x)
\end{bmatrix}
}
{q_1^TBq_1 + q_2^TBq_2 + 2\sqrt{q_1^TB(q_1q_2^T-q_2q_1^T)Bq_2}}
$$
where $B = A_xFF^TA_x^T$. Standard Beltrami coefficient $\hat{\mu}: \mathcal{D} \to \mathbb{C}$ is defined by
$$
\hat{\mu}(x)=
\frac{
\begin{bmatrix}
q_1^T & q_2^T
\end{bmatrix}
\begin{bmatrix}
B & iB\\
iB & -B
\end{bmatrix}
\begin{bmatrix}
q_1\\q_2
\end{bmatrix}
+
q^T
\begin{bmatrix}
\partial_1A_x & \partial_2A_x
\end{bmatrix}
\begin{bmatrix}
FF^T & iFF^T\\
iFF^T & -FF^T
\end{bmatrix}
\begin{bmatrix}
(\partial_1A_x)^T\\
(\partial_2A_x)^T
\end{bmatrix}q
}
{q_1^TBq_1 + q_2^TBq_2 + D_1 + 2\sqrt{q_1^TB(q_1q_2^T-q_2q_1^T)Bq_2 + D_2}}
$$
where
\begin{align*}
&D_1 = q^T((\partial_1A_x)FF^T(\partial_1A_x)^T+ (\partial_2A_x)FF^T(\partial_2A_x)^T)q\\
&D_2 = q^T(\partial_1A_x)FF^T((\partial_1A_x)^Tqq^T(\partial_2A_x) - (\partial_2A_x)^Tqq^T(\partial_1A_x))FF^T(\partial_2A_x)^Tq
\end{align*}

Discrete diffuse Beltrami coefficient \bm{$\tilde{\mu}$} is defined by \bm{$\tilde{\mu}$} $ = [\tilde{\mu}(p_1),\tilde{\mu}(p_2),...,$ $\tilde{\mu}(p_N)]^T$.
Similarly, discrete standard Beltrami coefficient \bm{$\hat{\mu}$} is defined by \bm{$\hat{\mu}$} $ = [\hat{\mu}(p_1), \hat{\mu}(p_2),...,\hat{\mu}(p_N)]^T$.
\end{defn}

The PCBCs defined above can measure how close a PC map from a 2D domain to a Riemann surface is to a conformal map. A PC conformal parameterization can be defined as follows.

\begin{defn}\label{Def:3dparam}
Let $\mathcal{P}$ be a point cloud surface, and $\bold{f}: \mathcal{P}\to \mathcal{D}$ be a injective point cloud map, where $\mathcal{D}$ is a simply connected compact domain in $\mathbb{R}^2$. Assume that $\mathcal{P}$ satisfies the interior cone condition. $\bold{f}^{-1}$ is called a parameterization of $\mathcal{P}$ if $\bold{f}(\mathcal{P})$ is a feasible point cloud in $\mathcal{D}$ and $\mathcal{P}$ is $d-$point cloud surface, where $d(p) = \max_{x\in \bold{f}^{-1}(S)}\|x-p\|$ and $S = B_{\delta}(\bold{f}(p))\cap \mathcal{P}$. $\bold{f}^{-1}$ is called a $e-$conformal parameterization if it is a parameterization and $\max\{\|\tilde{\mu}\|_\infty, \|\hat{\mu}\|_\infty\}$ $\leq e$, where $\tilde{\mu}$ and $\hat{\mu}$ are the diffuse and standard Beltrami coefficient of $\bold{f}^{-1}$.
\end{defn}

In the rest of this part, we assume that any PC is a PC surface sampled from a simply connected compact open Riemann surface, associated with a PC $e-$conformal parameterization map.

The PCBR associated to a PCQC map between two PC surfaces can now be defined.

\begin{defn}\label{Def:3dBelrep}
Let $\mathcal{P}_1$ and $\mathcal{P}_2$ be two point cloud surfaces with the same number of points, and $\bold{f}: \mathcal{P}_1\to \mathcal{P}_2$ be a bijective point cloud map. Let $\bm{\phi}_1$, $\bm{\phi}_2$ be two point cloud maps such that $\bm{\phi}_1^{-1}$ and $\bm{\phi}_2^{-1}$ are $e-$conformal parameterizations of $\mathcal{P}_1$ and $\mathcal{P}_2$ respectively. Then, the diffuse and standard PC Beltrami representation(PCBR) are defined by $\tilde{\mu}( \bm{\phi}_2\circ \bold{f} \circ \bm{\phi}_1^{-1})\circ \bm{\phi_1}$ and $\hat{\mu}( \bm{\phi}_2\circ \bold{f} \circ \bm{\phi}_1^{-1})\circ \bm{\phi_1}$, where $\tilde{\mu}(\cdot)$ and $\hat{\mu}(\cdot)$ denote the diffuse and standard Beltrami coefficient.
\end{defn}

In other words, the PCBR is defined by the PCBC of the projected map between the 2D conformal parameter domains. In practice, one can map each PC surface to the 2D domain using a $e-$conformal parameterization, on which the PCBR can be easily computed.

Next, we show that, under suitable conditions, a PC $e-$conformal parameterization is close to the actual conformal parameterization of a Riemann surface. With this observation, it follows that our proposed PCQC theories on PC surfaces are analogous to their continuous counterpart, up to an error controlled by fill distances of the PC surfaces.

\begin{prop}\label{Prop:3dconferr}
Let $\mathcal{S}\subseteq \mathbb{R}^3$ be a Riemann surface with global parameterization $\phi_0:\mathcal{D}\to \mathcal{S}$. Let $\mathcal{P}_n$ be a sequence of quasi-uniform point cloud surfaces sampled from $\mathcal{S}$ with $e_n-$conformal parameterization $\bm{\phi}_n$ and fill distance $h_{n,0}$ where $e_n$ and $h_{n,0}$ both converge to $0$, and $\mathcal{P}_n\subseteq \mathcal{P}_{n+1}$ for each $n$. Assume that $\lim_n\bm{\phi}_n^{-1}(x) = f(x)$ if $x\in \cup \mathcal{P}_n$. Further assume that $\bm{\phi}_n^{-1}(\mathcal{P}_n)$ and $f(\mathcal{P}_n)$ are feasible point clouds in $\mathcal{D}$ with fill distance $h_{n,1}$ and $h_{n,2}$, both of which converge to $0$. Let $\hat{\phi}_n$ be the MLS approximation of $\bm{\phi}_n$. Assume that $\partial_i\hat{\phi}_n$ is uniformly convergent for $i=1,2$, $\hat{\phi}_n$ converges to a smooth function $\phi$, and $\mu(\phi)$ is well-defined on $\mathcal{D}$. Then $\phi_0=\phi$ under suitable boundary condition.
\end{prop}
\begin{proof}
First, we prove that $\phi$ is a function from $\mathcal{D}$ to $S$. Let $p\in \mathcal{P}_n$, and $x_n = \bm{\phi}_n^{-1}(p)$, $x = \lim_{n} x_n = f(p)$, then
\begin{align*}
|\phi(x) - p|
&\leq |\phi(x) - \phi(x_n)| + |\phi(x_n) - \hat{\phi}_n(x_n)| + |\hat{\phi}_n(x_n) - p|\\
&\leq \sqrt{6}\|\phi\|_{C^1(\mathcal{D})}|x-x_n| + |\phi(x_n) - \hat{\phi}_n(x_n)| + |\hat{\phi}_n(x_n) - p|
\end{align*}
Without loss of generality, let $\varphi_i$ be the shape functions of $\bm{\phi}_n$, and denote $p=p_1$, then
$
|\hat{\phi}_n(x_n) - p|
= \left|\sum p_i\varphi_i(x_n) - p_1\right|
= \left|\sum \varphi_i(x_n)\cdot (p_i-p_1)\right|$.
Since $\varphi_i$ is only nonzero when $|x-p_i|\leq \delta_n$, then $|\hat{\phi}_n(x_n) - p|\leq \left(\sum|\varphi_i(x_n)|\right)\delta_n \leq C_1 \delta_n$, where $C_1$ is a constant independent of point cloud. Hence $\phi(f(p)) = p$ for arbitrary point $p\in \cup \mathcal{P}_n$. And by assumption, $f(\cup \mathcal{P}_n)$ is dense in $\mathcal{D}$, and $\phi$ is smooth, then $\phi(\mathcal{D})\subseteq \mathcal{S}$. On the other hand, since $h_{n,0}$ converges to $0$, we have $\cup \mathcal{P}_n$ dense in $\mathcal{S}$, hence $\phi(\mathcal{D})= \mathcal{S}$.
\\
Since $\partial_i\hat{\phi}_n$ is uniformly convergent, and $\hat{\phi}_n$ converges to a smooth function $\phi$, then $\partial_i \phi = \lim_n \partial_i\hat{\phi}_n$. By assumption, $\lim_n e_n=0$ and $\mu(\phi)$ is well-defined, then $0 = \lim_n\mu(\hat{\phi}_n) = \mu(\phi)$, and $\phi$ is diffeomorphism. By Riemann mapping theorem, $\phi_0$ is unique under suitable boundary condition, hence $\phi_0=\phi$.
\qed
\end{proof}

Suppose the PC $e-$conformal parameterizations are accurate enough, that is, they are close enough to the actual conformal parameterizations of the Riemann surfaces. In such a case, the following proposition states that the PCBR is close to the BC of $\phi_2^{-1}\circ f\circ \phi_1$, where $\phi_i$ is the actual conformal parameterization of Riemann surface $\mathcal{S}_i$ for $i=$1 or 2.

\begin{prop}\label{Prop:3dmuerr}
Let $f:\mathcal{S}_1\to \mathcal{S}_2$ be a quasi-conformal map, where Riemann surface $\mathcal{S}_j$ has global parameterization $\phi_j:\mathcal{D}\to \mathcal{S}_j$. Let $\mu_0 = \mu(\phi_2^{-1}\circ f\circ \phi_1)\circ \phi_1^{-1}$. Let $\mathcal{P}_1$ and $\mathcal{P}_2=f(\mathcal{P}_1)$ be two point clouds sampled from $\mathcal{S}_1$ and $\mathcal{S}_2$ with conformal parameterizations $\bm{\varphi}_1$ and $\bm{\varphi}_2$ respectively. Let $h_1$ be the fill distance of $\bm{\varphi}_1^{-1}(\mathcal{P}_1)$. Assume that $|\bm{\varphi}_i^{-1}(p)-\phi_i^{-1}(p)|\leq \epsilon$ for all $p\in \mathcal{P}_i$, $i=1,2$, and $\epsilon=o(h_1)$. Let $\tilde{\mu}$ be the diffuse Beltami representation of $f|_{\mathcal{P}_1}$. Then $\max_{p\in \mathcal{P}_1}|\tilde{\mu}(p) - \mu_0(p)| = O(h_1^2 + \frac{\epsilon}{h_1})$.
\\Similar result holds for standard Beltrami representation.
\end{prop}
\begin{proof}
Fix $p\in \mathcal{P}_1$. Let $x_0 = \phi_1^{-1}(p)$ and $x_1 = \bm{\varphi}_1^{-1}(p)$. Denote $\mu_1 = \mu(\phi_2^{-1}\circ f\circ \phi_1)$. Let $\tilde{\sigma}_1$ and $\tilde{\sigma}_2$ be the diffuse Beltrami coefficient of $\bm{\varphi}_2^{-1}\circ f \circ \bm{\varphi}_1$ and $\phi_2^{-1}\circ f\circ \phi_1|_{\bm{\varphi}_1^{-1}(\mathcal{P}_1)}$. Then
\begin{align*}
&|\tilde{\mu}(p)-\mu_0(p)|\\
\leq\ &|\tilde{\sigma}_1(x_1) - \tilde{\sigma}_2(x_1)| + |\tilde{\sigma}_2(x_1) - \mu_1(x_1)| + |\mu_1(x_1) - \mu_1(x_0)|\\
\leq\ &|\tilde{\sigma}_1(x_1) - \tilde{\sigma}_2(x_1)| + C_0h_1^2 + 2\|\mu_1\|_{C^1(\mathcal{D})}\epsilon
\end{align*}
Consider MLS on $\bm{\varphi}_1^{-1}(\mathcal{P}_1)$, let $a_{k,j}(y) = (q_k(y)^TA_y)_j$, $k=1,2$. By definition,
$$
\tilde{\sigma}_k = \frac{\left(\sum_j u_{k,j}a_{1,j} - \sum_j v_{k,j}a_{2,j}\right) +i\left(\sum_j v_{k,j}a_{1,j} + \sum_j u_{k,j}a_{2,j}\right)}{\left(\sum_j u_{k,j}a_{1,j} + \sum_j v_{k,j}a_{2,j}\right) +i\left(\sum_j v_{k,j}a_{1,j} - \sum_j u_{k,j}a_{2,j}\right)}
$$
where $[u_{1,j},v_{1,j}]^T = \bm{\varphi}_2^{-1}(f(p_j))$, and $[u_{2,j},v_{2,j}]^T = \phi_2^{-1}\circ f\circ \phi_1\circ \bm{\varphi}_1^{-1}(p_j)$. Hence
\begin{align*}
&|[u_{1,j},v_{1,j}]^T - [u_{2,j},v_{2,j}]^T|\\
\leq\ &|\bm{\varphi}_2^{-1}(f(p_j))-\phi_2^{-1}(f(p_j))| + |\phi_2^{-1}\circ f(p_j) - \phi_2^{-1}\circ f\circ \phi_1(\bm{\varphi}_1^{-1}(p_j))|\\
\leq\ &\epsilon + 2\|\phi_2^{-1}\circ f\circ \phi_1\|_{C^1(\mathcal{D})}\epsilon=:C_1\epsilon
\end{align*}
By MLS theory, $\sum_j |a_{k,j}(y)| \leq C_2h_1^{-1}$. Then $
|\sum_j u_{1,j}a_{k,j} - \sum_j u_{2,j}a_{k,j}|
\leq C_1C_2\epsilon/h_1
$, and $|\sum_j v_{1,j}a_{k,j} - \sum_j v_{2,j}a_{k,j}|\leq C_1C_2\epsilon/h_1$. Moreover, from error analysis of MLS, $\left|[\sum_j u_{2,j}a_{k,j},\sum_j v_{2,j}a_{k,j}]^T - \partial_k(\phi_2^{-1}\circ f\circ \phi_1) \right|\leq C_3h^2$.
By similar argument as in proof of Proposition \ref{Prop:muerr},
\begin{align*}
|\tilde{\sigma}_1(x_1) - \tilde{\sigma}_2(x_1)|
\leq \frac{16C_1C_2\epsilon\left(|\partial_z(\phi_2^{-1}\circ f\circ \phi_1)| + |\partial_{\bar{z}}(\phi_2^{-1}\circ f\circ \phi_1)|\right)}{L^2h_1}=:\frac{C_4\epsilon}{h_1}
\end{align*}
where $L = \min_{\mathcal{D}} |\partial_z (\phi_2^{-1}\circ f\circ \phi_1)|>0$.
Therefore,
$$
\max_{p\in \mathcal{P}_1}|\tilde{\mu}(p)-\mu_0(p)|\leq C_5\frac{\epsilon}{h_1} + C_0h_1^2
$$
where $C_5 = C_4 + 2\|\mu_1\|_{C^1(\mathcal{D})}h_0$.
\qed
\end{proof}

In other words, suppose the PC map is the restriction to the PC of a continuous surface QC map $f$. Under certain conditions, the PCBR is approximately equal to the Beltrami representation of $f$ restricted to the PC, up to an error related to the fill distances. The above proposition tells us that we can approximate the continuous Beltrami representation by computing the PCBR. The accuracy of the approximation can be improved as denser PCs are used to approximate the Riemann surfaces.

With the above error analysis, one can easily derive propositions about the PCBR as in the 2D case. The proof is almost the same as in the 2D case. We will omit the details but describe the main idea of the proof.

\begin{prop}\label{Prop:3dangle}
Let $f:\mathcal{S}_1\to \mathcal{S}_2$ be a quasi-conformal map. Assume $\mathcal{S}_i$ has global parameterization $\phi_i$.
Let $\{\mathcal{P}_{n,1}\}$ be a sequence of point clouds sampled from $\mathcal{S}_1$, and $\bold{f}_n$ be the corresponding point cloud map on $\mathcal{P}_{n,1}$. Denote $\mathcal{P}_{n,2} = f(\mathcal{P}_{n,1})$.
Assume $\mathcal{P}_{n,i}$ has conformal parameterization $\bm{\varphi}_{n,i}$ and $|\bm{\varphi}_{n,i}^{-1}(x)-\phi_i^{-1}(x)|\leq \epsilon_n$ for all $x\in \mathcal{P}_{n,i}$, $i=1,2$. Let $\bm{\varphi}_{n,1}(\mathcal{P}_{n,1})$ be feasible point cloud with fill distance $h_n$, and $h_n^2+\epsilon_n/h_n$ converges to $0$. Assume the sup-norm of diffuse or standard Beltrami representation converges to $0$. Then
$$
\max_{p_0\in \mathcal{P}_{n,1} \atop p_1,p_2\in T_n(p_0)\backslash \{p_0\}} \left|\alpha(\bold{f}_n(p_1),\bold{f}_n(p_0),\bold{f}_n(p_2)) - \alpha(p_1,p_0,p_2) \right| = O(h_n)
$$
where $T_n(p_0) = \bm{\varphi}_{n,1}(B_{\delta_n}(\bm{\varphi}_{n,1}^{-1}(p_0)))\cap \mathcal{P}_{n,1}$ and $\delta_n = C_\delta h_n$, $\alpha(p_1,p_0,p_2)$ denotes the angle between $p_1-p_0$ and $p_2-p_0$.
\end{prop}
\begin{proof}
The proof is essentially the same as the 2D case. We will omit the proof here.
\end{proof}

Again, the above proposition is analogous to the fact that angles are preserved under a surface conformal map. It states that the local angle structure of the PC surface is well preserved under the PCQC map, if the norm of its PCBR is small. Hence, one can again obtain a PC map that minimizes local geometric distortions, by controlling the PCBR norm.

On the other hand, PCBR can also measure local changes in statistical properties, such as local covariance matrices, under the PCQC map.

\begin{prop}\label{Prop:3dpca}
With the same assumptions as in Proposition \ref{Prop:3dmuerr}, furthermore assume $\bm{\varphi}_1^{-1}(\mathcal{P}_1)$ has quasi-uniform constant $c_{qu}/2$. Let $p$ be an arbitrary point in $\mathcal{P}_1$ such that $B_\delta(p')\subseteq \mathcal{D}$ where $p' = \bm{\varphi}_1^{-1}(p)$ and $\delta = C_\delta h_1$. Let $M_1$ be the covariance matrix for $S= \bm{\varphi}_1(B_\delta(p'))\cap \mathcal{P}_1$, and $M_2$ be the covariance matrix for $f(S)$. Assume matrix $M_i$ has three positive eigenvalues $\lambda_{i,1} \geq \lambda_{i,2}\geq \lambda_{i,3}$, and eigenvectors $v_{i,j}$ corresponding to $\lambda_{i,j}$. Assume $h_1 = o(\lambda_{1,1}/\lambda_{1,2}-1)$.
\\
(a). Let $w=[w_1,w_2]^T$ be the solution of $\nabla \phi_1(p') w = v_{1,2}$ in the least square sense. Denote $g(h_1) = h_1+\frac{\epsilon}{h_1} + \frac{h_1}{\lambda_{1,1}/\lambda_{1,2}-1}$. Then $|\lambda_{i,3}|=O(h_1^3)$, $i=1,2$, and
$$
\left|\frac{\lambda_{2,1}}{\lambda_{2,2}} - \left(\frac{1+|T|}{1-|T|}\right)^2\right| = O\left(g(h_1)\right)
$$
where
$$
T = \frac{(\sqrt{\lambda_{1,1}}+\sqrt{\lambda_{1,2}})\tilde{\mu}(p) - (\sqrt{\lambda_{1,1}}-\sqrt{\lambda_{1,2}})(w_1+iw_2)^2}{\sqrt{\lambda_{1,1}}+\sqrt{\lambda_{1,2}}-(\sqrt{\lambda_{1,1}}-\sqrt{\lambda_{1,2}})(w_1-iw_2)^2\tilde{\mu}(p)}
$$
\\
(b). Let $u_{0,j} = [\cos(\theta_j),\sin(\theta_j)]^T$, where $\theta_1 = angle(T)/2$, and $\theta_2 = \theta_1 +\pi/2$. Assume $v_{2,j}\cdot df_p(u_j)\geq 0$, where $u_j = \mathbb{P}((M_1+\sqrt{\lambda_{1,1}\lambda_{1,2}}I)\nabla \phi_1(p')u_{0,j})$, and $\mathbb{P}$ is the projection operator from $\mathbb{R}^3$ to $T_p\mathcal{S}_1$. Let $\zeta=\mu(p)-\frac{\sqrt{\lambda_{1,1}}-\sqrt{\lambda_{1,2}}}{\sqrt{\lambda_{1,1}}+\sqrt{\lambda_{1,2}}}(w_1+iw_2)^2$, and $g(h_1) = o(|\zeta|)$. Then
$$
\left|v_{2,j} - \frac{df_p(u_j)}{\|df_p(u_j)\|}\right|=O\left(h_1+\frac{g(h_1)}{|\zeta|}\right)
$$
\end{prop}
\begin{proof}

Let $\mathcal{P}_3 = \phi_1^{-1}(\mathcal{P}_1)$,
$M_3$ be the covariance matrix of $\phi_1^{-1}(S)$, and $M_4$ be the covariance matrix of $\phi_2^{-1}(f(S))$.
Assume $M_{i}$ has two eigenvalues $\lambda_{i,1}\geq \lambda_{i,2}$, and eigenvectors $v_{i,j}$ corresponding to $\lambda_{i,j}$, for $i=3,4$.
Denote $x_1 = \phi_1^{-1}(p)$, and $x_2 = \phi_2^{-1}(f(p))$.
Here we assume that $\epsilon\leq h_1$.

Then we can choose $\epsilon/h_1$ small enough such that quasi-uniform constant of $\mathcal{P}_3$ is $(h_1+\epsilon)/(2h_1/c_{qu}-\epsilon)\leq c_{qu}$, hence $\mathcal{P}_3$ is a feasible point cloud with fill distance $h_3\leq h_1+\epsilon\leq 2h_1$.

By direct calculation, we have $\lambda_{3,i}=\Theta(h_1^2)$, $|\lambda_{3,1}/\lambda_{3,2}-\lambda_{1,1}/\lambda_{1,2}|=O(h_1)$, and $|\lambda_{1,3}|=O(h_1^3)$. For $j=1,2$,
$$
\left|v_{1,j} - \frac{\nabla \phi_1(x_1)v_{3,j}}{\|\nabla \phi_1(x_1)v_{3,j}\|}\right| = O\left(h_1+\frac{h_1^3}{\lambda_{3,1} - \lambda_{3,2}}\right)
$$
Similarly,
$\lambda_{4,i} = \Theta(h_1^2)$, and
$
\frac{\lambda_{4,1}}{\lambda_{4,2}} = \left(\frac{1+|T_1|}{1-|T_1|}\right)^2 + O(h_1)
$, where
$$
T_1 = \frac{(\sqrt{\lambda_{3,1}}+\sqrt{\lambda_{3,2}})\mu_0(p) - (\sqrt{\lambda_{3,1}}-\sqrt{\lambda_{3,2}})((v_{3,2})_1+i(v_{3,2})_2)^2}{\sqrt{\lambda_{3,1}}+\sqrt{\lambda_{3,2}}-(\sqrt{\lambda_{3,1}}-\sqrt{\lambda_{3,2}})((v_{3,2})_1-i(v_{3,2})_2)^2\mu_0(p)}
$$
Moreover, for $j=1,2$, let $u'_{0,j} = [\cos(\theta'_j), \sin(\theta'_j)]^T$, where $\theta'_1 = angle(T_1)/2$, and $\theta'_2 = \theta'_1+\pi/2$, then
$$
\left|v_{4,j} - \frac{\nabla (\phi_2^{-1}\circ f\circ \phi_1)|_{x_0}(M_3 + \sqrt{\lambda_{3,1}\lambda_{3,2}}I)u'_{0,j} }{\|\nabla (\phi_2^{-1}\circ f\circ \phi_1)|_{x_0}(M_3 + \sqrt{\lambda_{3,1}\lambda_{3,2}}I)u'_{0,j}\|}\right| = O\left(h_1+\frac{h_1}{|\zeta|} \right)
$$
For the map $\phi_2$, we have similar result, $|\lambda_{2,1}/\lambda_{2,2}-\lambda_{4,1}/\lambda_{4,2}|=O(h_1)$, and $|\lambda_{2,3}|=O(h_1^3)$. For $j=1,2$,
$$
\left|v_{2,j} - \frac{\nabla \phi_2(x_2)v_{4,j}}{\|\nabla \phi_2(x_2)v_{4,j}\|}\right| = O\left(h_1+\frac{h_1^3}{\lambda_{4,1} - \lambda_{4,2}}\right)
$$
Together with Proposition \ref{Prop:3dmuerr}, the equation in above statement can be obtained.
\qed
\end{proof}

Again, Proposition 16a is analogous to the fact that the dilation of the infinitesimal ellipse from a deformed infinitesimal circle under the QC map can be determined from the Beltrami differential.

We have now constructed the PCQC theories on PC surfaces sampled from simply connected Riemann surfaces. All the above definitions and statements only depend on local properties of the PC map. Hence, the theories can be easily generalized to PC surfaces sampled from general Riemann surfaces with arbitrary topologies, such as multiply-connected open or high genus closed surfaces. 

According to our developed theories, every PCQC map between two PC surfaces can be represented by its PCBR, which captures local geometric distortions of the deformation. With these concept of computational QC geometry, many existing QC based algorithms for geometry processing and shape analysis can be easily readily to PC data.

\section{Experimental Result}
We have numerically validated our proposed theories described in the last section. In this section, the numerical results are reported.

\subsection{Experiment on the choice of weight functions}
The MLS method is adopted in this paper to approximate partial derivatives on PCs. Several weight functions are often used in MLS methods, such as the Gauss function, Wendland function and cubic function. In this subsection, we compare these three weight functions and choose one for our experiments in the rest of this section. The Gauss, Wendland and cubic functions, denoted by $W_1$, $W_2$, and $W_3$ respectively, are given below:
\begin{align*}
W_1(x)&=\exp\left(-\frac{\delta^2 x}{h^2}\right)\\
W_2(x)&=(1-x)^4(1+4x)\\
W_3(x) &= \begin{cases} 2/3 - 4x^2 + 4x^3 & \quad \text{if } x <0.5\\
4/3-4x + 4x^2 - 4x^3/3 & \quad \text{if } x\geq 0.5\\ \end{cases}
\end{align*}
where $h$ is the separation distance and $\delta$ is radius of neighbourhood, which is chosen as $6h$ in our numerical experiments. To avoid the non-differentiablility on the central point, the weight function is chosen to be $w(d) = W_i(d^2)$ for each cases.

The test function we use is $f(x,y) = (x^2+1)sin(y)$.
We approximate both the test function and its first order partial derivatives using the MLS method with different weight functions. In each case, we denote the approximation of $f$ by $\hat{f}$, and the approximation of $\partial_jf$ by $\hat{f}_j$. The numerical error $e = \max_{j}\{\|f-\hat{f}\|_1, \|\partial_jf-\hat{f}_j\|_1\}$ is then computed. Note that the partial derivatives can either be approximated by the diffuse or standard derivatives.
The numerical error for each weight function is recorded in Table \ref{tab:weight}.

\bgroup
\def\arraystretch{1.2}
\begin{table}
\center
\begin{tabular}{|c|c|c|c|c|c|c|}
\hline
\multirow{2}{*}{No. pts} &
\multicolumn{3}{ |c| }{Diffuse derivative} &
\multicolumn{3}{ |c| }{Standard derivative} \\
\cline{2-7}
&Gauss & Wendland & Cubic & Gauss & Wendland & Cubic  \\
\hline
$25^2$&	2.70E-4&	4.67E-4&	4.99E-4&	2.56E-4&	9.72E-4&	1.23E-3\\
$33^2$&	1.52E-4&	2.64E-4&	2.82E-4&	1.44E-4&	5.49E-4&	6.96E-4\\
$49^2$&	6.79E-5&	1.18E-4&	1.26E-4&	6.43E-5&	2.45E-4&	3.10E-4\\
$65^2$&	3.82E-5&	6.67E-5&	7.13E-5&	3.62E-5&	1.38E-4&	1.75E-4\\
$97^2$&	1.70E-5&	2.97E-5&	3.17E-5&	1.61E-5&	6.17E-5&	7.79E-5\\
$129^2$&	9.58E-6&	1.67E-5&	1.79E-5&	9.07E-6&	3.47E-5&	4.38E-5\\
\hline
\end{tabular}
\caption{Error of diffuse and standard derivatives using MLS approximation with different weight functions}
\label{tab:weight}       
\end{table}
\egroup

From Table \ref{tab:weight}, it can be observed that the Gauss weight function performs better than other two weight functions for the test function $f$. Therefore, we will use the Gauss weight in the rest of our numerical experiments.

\subsection{Numerical validations of propositions}
In this subsection, we numerically validate the propositions proved in the last section. In our experiments, all input PCs are quasi-uniform. Hence, we use the $k$-NN neighborhood, instead of the disk neighborhood with fixed radius, to compute diffuse and standard PCBCs, where $k$ is set to be $25$. In order to show the convergence rates of errors in these propositions, we use a straight line to fit the data points in each graph and translate the line for easier comparison. The fitting lines after translation are plotted with red dash lines, whose slopes indicate the corresponding convergence rates. The actual error plots are shown with blue lines.

\subsubsection{Numerical validation of Proposition \ref{Prop:muerr}, \ref{Prop:Pca}, \ref{Prop:3dmuerr} and \ref{Prop:3dpca}}
In this part, we numerical validate Proposition \ref{Prop:muerr}, \ref{Prop:Pca}, \ref{Prop:3dmuerr} and \ref{Prop:3dpca} on 2D and 3D PCs respectively.

For the 2D case, we use the function $f(x,y) = [e^x, (x^2+1)sin(y)]^T$ as the underlying QC map. A sequence of PCs of decreasing fill distances are used to compute the PCBCs. The PCBCs are compared with the actual Beltrami coefficient of $f$ to compute convergence rates. In this experiment, $\mathcal{D}$ is chosen to be the unit rectangle. Each PC is union of vertices taken from a rectangle grid with length $h$. More specifically, we subdivide the unit rectangle to several small squares. The corresponding PC contains vertices of all squares. Note that the fill distance is proportional to the grid size. Thus, we use the grid size $h$ in our numerical error analysis, instead of the fill distance.

The numerical results are shown in Figure \ref{fig:pca-2d}. (A) and (B) show the coarsest PC data before and after the deformation by $f$, where the colormaps are given by the norm of the actual BC of $f$. (C) and (D) verify the error bounds of diffuse and standard PCBCs in Proposition \ref{Prop:muerr}. (E) and (F) verify the error bounds in Proposition \ref{Prop:Pca} (a) and (b) respectively. In each figure, $e$ is the sup-norm of error in the corresponding approximation. We plot $\log(e)$ versus $-\log(h)$ to show the convergence rate. In (C) and (D), we observe similar convergence rates as what we obtained in the propositions. On the other hand, (E) and (F) give convergence rates about 1.85 and 1.42, which are better than what we proved. It is mainly due to the regularity of the PC data. Moreover, the error of the standard PCBC is very close to error of the diffuse PCBC. In general, PCBCs computed by diffuse derivatives and standard derivatives are very similar. Hence, in the following experiments, we will only show results calculated using diffuse derivatives.

For the 3D case, we choose parameterization functions $\phi_1 = [cosh(x)cos(y),$ $cosh(x)sin(y), x]^T$ and $\phi_2 = [cos(y), sin(y), x]^T$. The Riemann surfaces in this experiment are therefore given by $\mathcal{S}_1 = \phi_1(\mathcal{D})$ and $\mathcal{S}_2 = \phi_2\circ f(\mathcal{D})$, where $\mathcal{D}$ and $f$ are the same as in our 2D experiment. To compute the PCBR, we need to compute PC $e-$conformal parameterizations for $\mathcal{S}_1$ and $\mathcal{S}_2$. Two algorithms for computing PC parameterization have been proposed by Zhao et al. in \cite{Hongkai3}, \cite{Hongkai1}, \cite{Hongkai2}. Here, we use the method reported in \cite{Hongkai2}, since the algorithm applied MLS approximation. We will now briefly describe the method. For each point, a local projection function is approximated, whose graph gives a local approximation of the surface. With this approximation, the MLS method is applied again to obtain the Laplace-Beltrami operator at that point. Suppose the PC surface satisfies the conditions as described in Definition \ref{Def:pcsurf}. The existence of a local injective PC projection map can be guaranteed. By solving the Laplace-Beltrami equation, a PC conformal parameterization can be obtained.

In this numerical experiments, we construct the input 3D PC data as follows. For each regular 2D PC $\mathcal{P}$ obtained in the last experiment, we choose $\mathcal{P}_1 = \phi_1(\mathcal{P})$ and $\mathcal{P}_2 = \phi_2\circ f (\mathcal{P})$ as our input 3D PC surfaces. Using the parameterization method in \cite{Hongkai2}, one can compute the PC parameterization $\bm{\varphi}_i$ for $\mathcal{P}_i$. The numerical errors of the PC parameterizations from the actual parameterizations, $\|\bm{\varphi}_i^{-1}(\mathcal{P}_i) - \phi_i^{-1}(\mathcal{P}_i)\|_\infty$, are shown in Figure \ref{fig:pca-3d} (A) and (B). Note that the numerical errors of PC parameterizations converge faster than linear convergence, hence the conditions in Proposition \ref{Prop:3dmuerr} and \ref{Prop:3dpca} are satisfied. The numerical errors as stated in Proposition \ref{Prop:3dmuerr} and \ref{Prop:3dpca} are computed. The results are shown in Figure \ref{fig:pca-3d} (C)-(I). We use the sup-norm error for the numerical analysis of diffuse and standard PCBRs. As for Proposition \ref{Prop:3dpca}, we consider the numerical error at one point only, because the values $\lambda_{1,1}/\lambda_{1,2}-1$ and $|\zeta|$ in the error bound are variant from point to point.

As one can observe from the graphs, the convergence rates of diffuse and standard PCBC are about second order, which validates the result of Proposition \ref{Prop:3dmuerr}. From the graphs (E)-(I), one can see that the ratio of first and second eigenvalues converges with a quadratic rate, and the eigenvectors also converge quadratically. Moreover, the third eigenvalues of both local covariance matrices $M_1$ and $M_2$ perform fourth order convergences.
These convergence rates are better than what we have proved, because of the regularity of the PC data.

\begin{figure}
  \includegraphics[width=\textwidth]{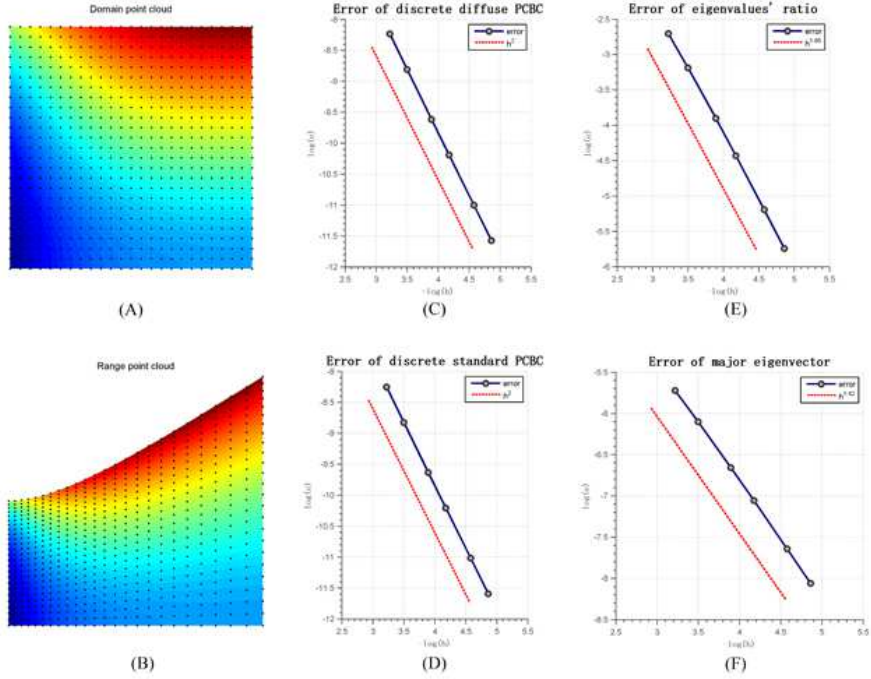}
\caption{Numerical tests of Propositions \ref{Prop:muerr}, \ref{Prop:Pca} on regular planar PCs. (A) and (B) show the coarsest data before and after the PCQC deformation. (C) and (D) show the errors of diffuse and standard PCBCs. (E) and (F) show the errors of eigenvalues' ratio in Proposition \ref{Prop:Pca} (a) and (b).}
\label{fig:pca-2d}       
\end{figure}

\begin{figure}
  \includegraphics[width=\textwidth]{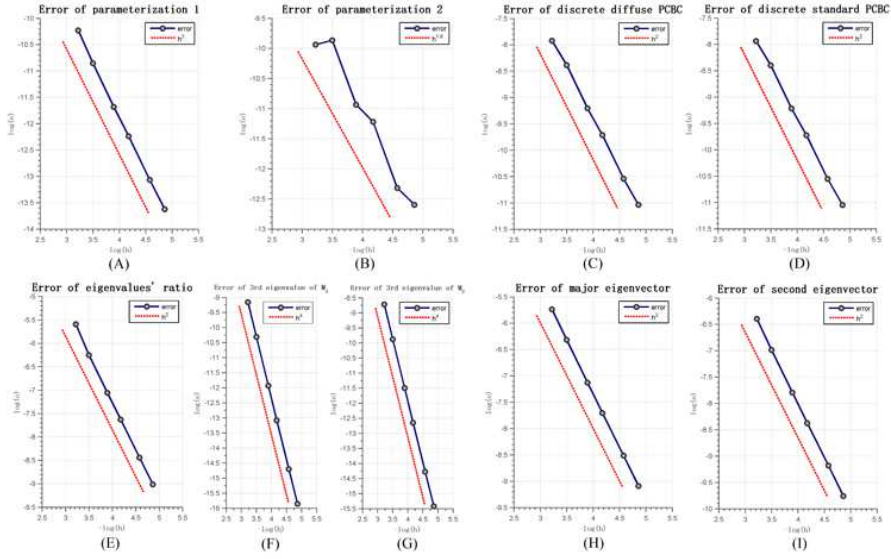}
\caption{Numerical tests of Propositions \ref{Prop:3dmuerr}, \ref{Prop:3dpca} on regular PC surfaces. (A) and (B) show the errors of PC conformal parameterizations. The errors of diffuse and standard PCBRs are given in (C) and (D). (E) shows the error of 1st and 2nd eigenvalues' ratio in Proposition \ref{Prop:3dpca} (a). The absolute values of 3rd eigenvalues of $M_1$ and $M_2$ are shown in (F) and (G). The errors of 1st and 2nd eigenvectors in Proposition \ref{Prop:3dpca} (b) are shown in (H) and (I).}
\label{fig:pca-3d}       
\end{figure}

\subsubsection{Proposition \ref{Prop:angle} and \ref{Prop:3dangle}}
In this part, we show two experiments for Propositions \ref{Prop:angle} and \ref{Prop:3dangle}. The test function is chosen to be $f(z) = 0.4(z+0.5i)^2 + 0.008(z-0.45-0.4i)^3+0.032(z-0.4-0.35i)^4$, while other setting remains the same as in last experiment.
The results are shown in Figure \ref{fig:ang-2d} and \ref{fig:ang-3d} for 2D and 3D cases respectively.

In Figure \ref{fig:ang-2d}, (A) and (B) present one PC data before and after mapping. We can see from (C) that the sup-norm of diffuse PCBC converges to $0$, which satisfies our condition. (D) shows the convergence rate of local angle changes.
In Figure \ref{fig:ang-3d}, the sup-norm errors of two PC parameterization maps are shown in (A) and (B). The convergence rates are faster than first order, which satisfies our assumption. Similar to 2D case, one can observe from (C) that sup-norm of diffuse PCBR converges to $0$, and (D) shows the convergence rate of local angle changes.

In these two tests, the slopes of best fitting lines have absolute values about 0.89 and 0.85 respectively, which are slightly smaller than 1. It is due to the unstable numerical errors of the PC conformal parameterization, as shown in Figure \ref{fig:ang-3d}(b). In order to observe the tendency of slopes when $h$ goes to $0$, we calculate the slopes of lines between each two consecutive data points. For 2D case, those slopes have absolute values about 0.81, 0.86, 0.89, 0.92, 0.94. For 3D case, the absolute values of slopes are computed to be 0.70, 0.79, 0.86, 0.90, 0.93. Therefore, in both cases, the slopes have a tendency to converge to 1. Moreover, the linear convergence is believed to be the best convergence rate.

\begin{figure}
  \includegraphics[width=\textwidth]{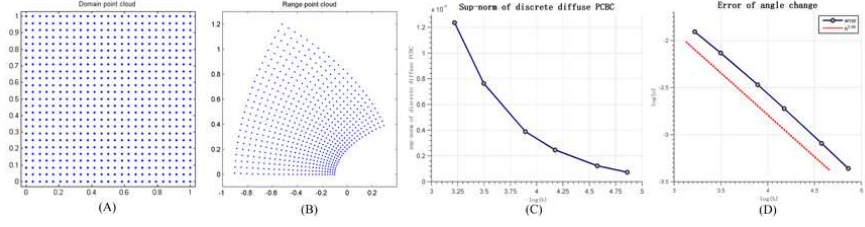}
\caption{Numerical test of Proposition \ref{Prop:angle} on regular planar PCs. (A) and (B) show the coarsest data before and after the PCQC deformation. (C) shows the sup-norm of diffuse PCBC. (D) shows the errors of local angle changes in Proposition \ref{Prop:angle}.}
\label{fig:ang-2d}       
\end{figure}

\begin{figure}
  \includegraphics[width=\textwidth]{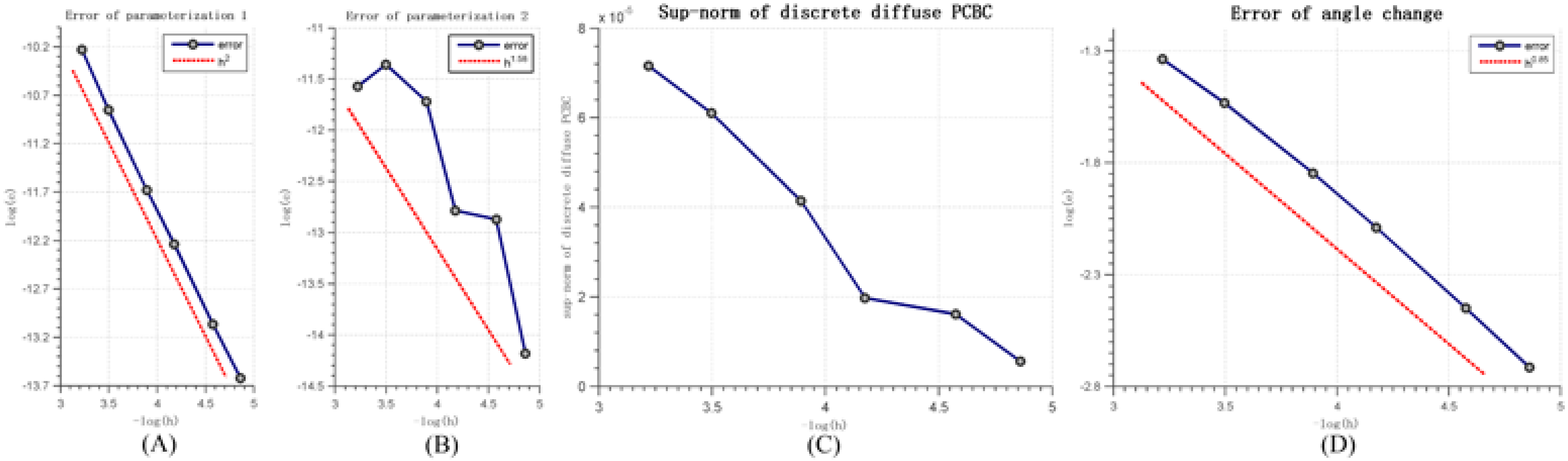}
\caption{Numerical test of Proposition \ref{Prop:3dangle} on regular PC surfaces. (A) and (B) show the error of PC conformal parameterizations. (C) shows the sup-norm of diffuse PCBR. (D) shows the errors of local angle changes in Proposition \ref{Prop:3dangle}.}
\label{fig:ang-3d}       
\end{figure}

\subsection{Quasi-conformal point cloud map calculated from Beltrami coefficient}
In this part, we demonstrate the effectiveness of representing PCQC map by its PCBC or PBBR. We obtain the PCQC map from a prescribed PCBC or PCBR using different meshless methods. To solve the PCQC map from the prescribed PCBC, one can solve either the Beltrami's equation \cite{Gardiner} or generalized Laplace equation \cite{LuiTMap}. 

The Beltrami's equation is given by $\partial_{\bar{z}} f = \mu(f)\partial_z f$, which is directly obtained from the definition of BC. This equation has a unique solution with suitable boundary condition. For the generalized Laplace system, let $\mu(f) = \sigma + i\tau$, then the PDE equation is given by $\nabla \cdot(A\nabla f) = 0$, where $A=\bigl(\begin{smallmatrix} a&b \\ b&c \end{smallmatrix} \bigr)$, and $a = \frac{1+|\mu|^2-2\sigma}{1-|\mu|^2}$, $b = \frac{-2\tau}{1-|\mu|^2}$, $c = \frac{1+|\mu|^2+2\sigma}{1-|\mu|^2}$. This equation can also be derived from definition, and it is of elliptic type since $\det(A)=1>0$. Furthermore, generalized Laplace equation also has unique solution under suitable boundary condition, and this PDE is just Laplace equation when $|\mu|=0$.

We use two discretization methods to solve these PDEs, including collocation method and element free Galerkin (EFG) method, which are two important meshless methods related to MLS approximation. For further information of these two methods, we refer the readers to \cite{EFG}, \cite{meshfreeintro}.

We have three experiments on three different functions given by
\begin{align*}
&f_1(x,y) = [e^x, (x^2+1)y]^T\\
&f_2(x,y) = [\sin(x), (x^2+1)\sin(y)]^T\\
&f_3(z) = 0.4(z+0.5i)^2 + 0.008(z-0.45-0.4i)^3+0.032(z-0.4-0.35i)^4
\end{align*}
 In the first example, we test each method on regular PCs sampled from rectangle grids, with exact solution $f_1$.
 In the second example, the input PCs are regular PCs with some random error, and the exact solution is function $f_2$.
 In the third example, we use test function $f_3$ on PCs consisting of vertices of triangular meshes.

For each computed solution $g_1$ and exact solution $g_0$ on PC, if $g_0$ is not zero function, the relative error $e$ is calculated by
$$
e = \frac{\sum |g_1(p_i) - g_0(p_i)|}{\sum |g_0(p_i)|}
$$
Otherwise, if $g_0(x) = 0$ for all $x$, we simply take the average of errors over every points.

For each experiment, the CPU time (in seconds) is given in Table \ref{tab:time}. The errors of numerical approximation for the PCQC maps from exact BCs, diffuse and standard PCBCs are recorded in Table \ref{tab:verr}, \ref{tab:dmuerr} and \ref{tab:smuerr} respectively.

\bgroup
\def\arraystretch{1.2}

\begin{table}
\center
\begin{tabular}{ |c|c|c|c|c|c| }
\hline
\multirow{2}{*}{No.} &
\multirow{2}{*}{pts num} &
\multicolumn{2}{ |c| }{Beltrami Equation} &
\multicolumn{2}{ |c| }{Generalized Laplace} \\ \cline{3-6}
& & collocation & EFG & collocation & EFG\\
\hline
\multirow{4}{*}{1}
&$24^2$&	0.183& 	6.444& 		0.590& 	4.917 \\
&$32^2$&	0.432& 	15.987& 	0.666& 	8.616 \\
&$48^2$&	1.141& 	35.707& 	1.610& 	19.381 \\
&$64^2$&	1.902& 	68.326& 	2.737& 	35.241 \\
\hline
\multirow{4}{*}{2}
&$24^2$&	0.186& 	7.554& 		0.360& 	4.237 \\
&$32^2$&	0.399& 	23.040& 	0.528& 	7.330 \\
&$48^2$&	0.931& 	31.314& 	1.335& 	16.399 \\
&$64^2$&	1.823& 	60.240& 	2.138& 	28.158 \\
\hline
\multirow{4}{*}{3}
&1047&	0.317& 	13.788& 	0.573& 	7.135 \\
&1807&	0.744& 	22.611& 	0.987& 	12.436 \\
&4132&	1.847& 	61.306& 	2.206& 	29.526 \\
&7185&	3.236& 	112.171& 	4.042& 	50.446 \\
\hline
\end{tabular}
\caption{CPU time (in second) to compute PCQC maps}
\label{tab:time}
\end{table}

\begin{table}
\center
\begin{tabular}{ |c|c|c|c|c|c| }
\hline
\multirow{2}{*}{No.} &
\multirow{2}{*}{pts num} &
\multicolumn{2}{ |c| }{Beltrami Equation} &
\multicolumn{2}{ |c| }{Generalized Laplace} \\ \cline{3-6}
& & collocation & EFG & collocation & EFG\\
\hline
\multirow{4}{*}{1}
&$24^2$&	3.037E-05&	2.020E-05&	3.706E-05&	3.263E-03\\
&$32^2$&	1.711E-05&	1.365E-05&	1.928E-05&	2.159E-03\\
&$48^2$&	7.668E-06&	3.056E-06&	7.926E-06&	9.497E-04\\
&$64^2$&	5.416E-06&	1.777E-06&	4.364E-06&	6.018E-04\\
\hline
\multirow{4}{*}{2}
&$24^2$&	2.588E-04&	1.928E-04&	2.545E-04&	1.372E-02\\
&$32^2$&	1.012E-04&	1.180E-04&	1.834E-04&	1.040E-02\\
&$48^2$&	3.772E-05&	3.790E-05&	5.061E-05&	6.994E-03\\
&$64^2$&	1.391E-05&	2.454E-05&	3.296E-05&	5.620E-03\\
\hline
\multirow{4}{*}{3}
&1047&	1.479E-05&	1.444E-05&	3.101E-06&	1.303E-01\\
&1807&	8.580E-06&	5.818E-06&	1.293E-06&	8.692E-02\\
&4132&	4.042E-06&	2.019E-06&	1.433E-06&	6.142E-02\\
&7185&	2.018E-06&	1.123E-06&	5.923E-07&	4.618E-02\\
\hline
\end{tabular}
\caption{Error of PCQC maps}
\label{tab:verr}
\end{table}

\begin{table}
\center
\begin{tabular}{ |c|c|c|c|c|c| }
\hline
\multirow{2}{*}{No.} &
\multirow{2}{*}{pts num} &
\multicolumn{2}{ |c| }{Beltrami Equation} &
\multicolumn{2}{ |c| }{Generalized Laplace} \\ \cline{3-6}
& & collocation & EFG & collocation & EFG\\
\hline
\multirow{4}{*}{1}
&$24^2$&	4.477E-05&	9.434E-04&	6.906E-04&	1.610E-01\\
&$32^2$&	2.384E-05&	5.332E-04&	3.839E-04&	1.287E-01\\
&$48^2$&	1.015E-05&	1.454E-04&	1.686E-04&	6.291E-02\\
&$64^2$&	7.254E-06&	1.087E-04&	9.599E-05&	4.055E-02\\
\hline
\multirow{4}{*}{2}
&$24^2$&	2.980E-04&	3.665E-03&	1.660E-03&	4.188E-01\\
&$32^2$&	1.664E-04&	2.732E-03&	1.458E-03&	4.190E-01\\
&$48^2$&	4.348E-05&	1.241E-03&	3.923E-04&	3.992E-01\\
&$64^2$&	1.703E-05&	1.053E-03&	2.593E-04&	3.977E-01\\
\hline
\multirow{4}{*}{3}
&1047&	9.719E-05&	9.882E-05&	1.907E-05&	1.934E+00\\
&1807&	7.198E-05&	5.451E-05&	1.007E-05&	1.691E+00\\
&4132&	5.201E-05&	2.922E-05&	4.797E-06&	2.028E+00\\
&7185&	3.491E-05&	2.005E-05&	2.813E-06&	1.840E+00\\
\hline
\end{tabular}
\caption{Error of diffuse PCBCs from computed PCQC maps}
\label{tab:dmuerr}
\end{table}

\begin{table}
\center
\begin{tabular}{ |c|c|c|c|c|c| }
\hline
\multirow{2}{*}{weight fn} &
\multirow{2}{*}{pts num} &
\multicolumn{2}{ |c| }{Beltrami Equation} &
\multicolumn{2}{ |c| }{generalized Laplace} \\ \cline{3-6}
& & collocation & EFG & collocation & EFG\\
\hline
\multirow{4}{*}{Example 1}
&$24^2$&	8.127E-06&	9.158E-04&	6.658E-04&	1.610E-01\\
&$32^2$&	3.554E-06&	5.182E-04&	3.708E-04&	1.286E-01\\
&$48^2$&	1.728E-06&	1.379E-04&	1.632E-04&	6.291E-02\\
&$64^2$&	4.039E-06&	1.045E-04&	9.305E-05&	4.055E-02\\
\hline
\multirow{4}{*}{Example 2}
&$24^2$&	2.847E-04&	3.616E-03&	1.657E-03&	4.149E-01\\
&$32^2$&	1.596E-04&	2.694E-03&	1.453E-03&	4.143E-01\\
&$48^2$&	3.902E-05&	1.222E-03&	3.914E-04&	3.946E-01\\
&$64^2$&	1.444E-05&	1.038E-03&	2.588E-04&	3.929E-01\\
\hline
\multirow{4}{*}{Example 3}
&1047&	1.432E-05&	2.506E-05&	2.056E-05&	3.822E-01\\
&1807&	9.558E-06&	1.325E-05&	1.092E-05&	3.267E-01\\
&4132&	6.291E-06&	7.047E-06&	5.152E-06&	3.651E-01\\
&7185&	4.240E-06&	5.386E-06&	3.024E-06&	3.597E-01\\
\hline
\end{tabular}
\caption{Error of standard PCBCs corresponding to computed PCQC maps}
\label{tab:smuerr}
\end{table}

\egroup

From the above experiments, we observe that the errors of numerical approximation for PCQC maps are smaller if the Beltrami's equation is solved using EFG method. However, the errors of the approximation from the diffuse and standard PCBCs are smaller when solving the generalized Laplace equation with the collocation method. In general, the errors of approximation for PCQC maps from diffuse PCBCs and standard PCBCs are quite small. It demonstrates that our proposed (diffuse or standard) PCBC provides a good representation of a PC deformation, which captures local geometric distortion of the deformation. Given a PCQC map, we can compute its diffuse or standard PCBC. Conversely, given a PCBC, we can solve for the associated PCQC map with tolerable errors.

Furthermore, solving the Beltrami's equation is more suitable for regular PCs, while solving generalized Laplace equation is more suitable for non-regular PCs. Therefore, in the next subsection, we will use the collocation method to solve generalized Laplace equation to obtain PCQC maps from the PCBCs, since PCs are mostly irregular in practice.

\subsection{Experiment on real data}
We have also performed experiments to solve for the PCQC map from a given PCBC. The results are shown in Figure \ref{fig:redbox} and \ref{fig:shapedball}. The PCs are downloaded from project Aim@Shape \cite{aimshape}. For each experiment, we give an artificial PCBR on each point. The PCQC map is then approximated on 2D parameter domain of the input BC. More precisely, the algorithm can be described below.

When a PC surface is given, the Laplace-Beltrami operator can be approximated on each point. In this work, we applied the algorithm in \cite{Hongkai2}. By solving the discrete Laplace-Beltrami equation, we can obtain a PC $e-$conformal parameterization. The PC surface is parameterized onto a 2D parameter domain. Then, we use collocation method to solve the generalized Laplace equation with the Gauss weight function to obtain the PCQC map.

In each of Figure \ref{fig:redbox} and \ref{fig:shapedball}, (A) shows the original PC surface. (B) shows PCQC parameterization of the input PC corresponding to the prescribed PCBR. (C) shows (in blue color) the histogram of the norm of the prescribed PCBR. We also compute the PCBR of the approximated PCQC parameterization. The histogram of the norm of the approximated PCBR is also shown (in red color). Note that the prescribed PCBRs closely resemble to  the approximated PCPRs in both cases. It demonstrates that the PCBR provides a good representation of a PC surface deformation. Given a PC surface deformation, we can compute its associated PCBR. Conversely, given a PCBR defined on each point, we can reconstruct a PCQC surface map, whose PCBR closely resembles to the prescribed PCBR. The representation of a PC deformation using its PCBR is especially useful since it captures the local geometric distortion. By incorporating the PCBR into the optimization model, an optimal PC deformation with minimal local geometric distortion can be obtained. It can be used for PC registration.

\begin{figure}
  \includegraphics[width=\textwidth]{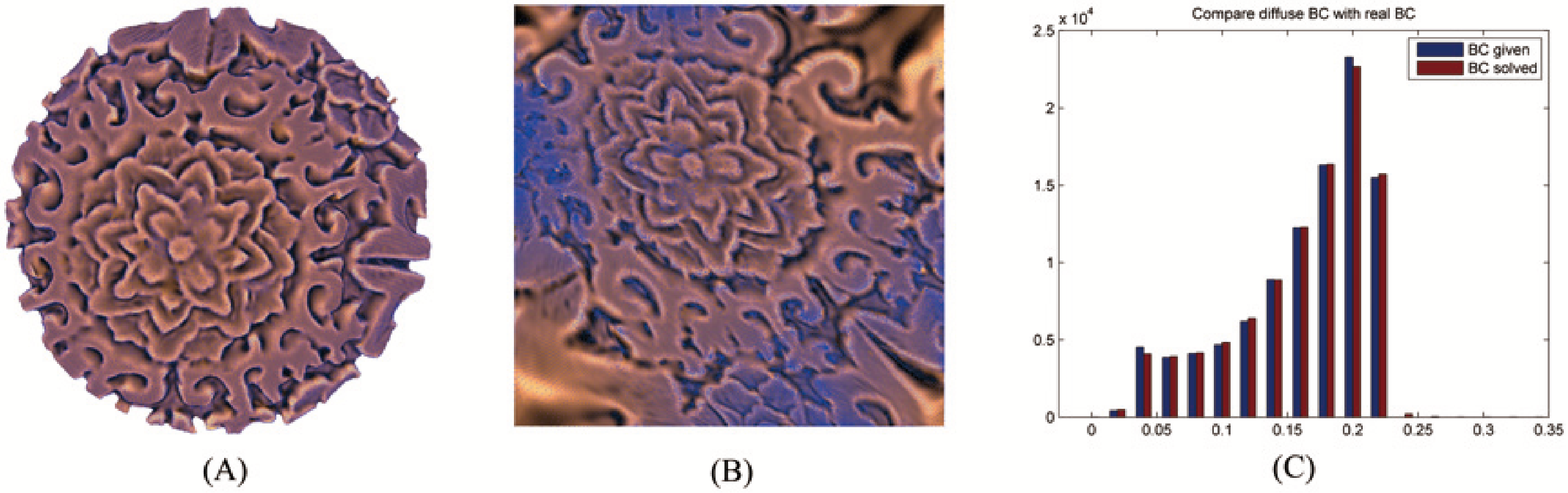}
\caption{PCQC parameterization with prescribed PCBR on a PC surface (redbox). (A) shows the original PC surface. (B) shows the PC parameterization with the prescribed PCBR. (C) compares the norms of the PCBR of the computed PCQC parameterization and the prescribed PCBR.}
\label{fig:redbox}       
\end{figure}

\begin{figure}
  \includegraphics[width=\textwidth]{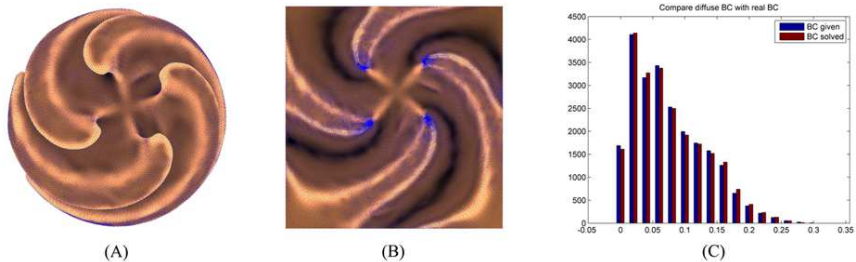}
\caption{PCQC parameterization with prescribed PCBR on a PC surface (shaped ball). (A) shows the original PC surface. (B) shows the PC parameterization with the prescribed PCBR. (C) compares the norms of the PCBR of the computed PCQC parameterization and the prescribed PCBR.}
\label{fig:shapedball}       
\end{figure}

\section{Conclusion}
In this paper, we develop the computational QC theories on PCs sampled from either simply connected planar domains or Riemann surfaces embedded in $\mathbb{R}^3$. The proposed theories can also be easily generalized to Riemann surfaces with arbitrary topologies. We define the concepts of PC quasi-conformal (PCQC) maps and their associated PC Beltrami coefficients (PCBCs), which are analogous to the Beltrami differential in the continuous setting. The PCBC converges to the continuous BC as the PC get denser and denser, under suitable conditions on the PCs. The PCBC also captures local geometric information of the PC deformation. We theoretically and numerically examine the ability of PCBCs to measure local geometric distortions under PC deformations. Extensive experiments on synthetic and real data have been carried out to validate our theoretical findings. In the future, we plan to apply our proposed theories to practical applications in computer graphics, medical imaging and computer vision.

\section{Appendix}
\subsection{Proof of Proposition \ref{Prop:Pca}(a)}
\begin{proof}
Firstly, one can calculate the range of $\lambda_1/\lambda_2$ by the following process. Let $S=\{p_i, \,i=1,...,n\}$, and $p_0 = (\sum_{i=1}^n p_i)/n$. By rotation with center $p_0$ and translation, we can obtain a new covariance matrix $M_3$ which is a diagonal matrix with $(M_3)_{11} = \lambda_1$ and $(M_3)_{22} = \lambda_2$. Let the new position of each point $p_i$ be $(x'_i,y'_i)$ and furthermore, $\sum_{i=1}^n (x'_i,y'_i)=(0,0)$. Then $\lambda_1 = (\sum_{i=1}^n (x'_i)^2)/n \leq \delta^2$, and $\lambda_2 \geq (\delta-2h)^2/n$ because there must be a point in the disk $B_h(0,\delta-h)$. Furthermore, since the $n$ disks $B_{q_\mathcal{P}}(p_i)$ are disjoint and inside $B_{q_\mathcal{P}+\delta}(0)$, by calculating their areas we can get
$$
n\leq \left( \frac{q_\mathcal{P} + \delta}{q_\mathcal{P}} \right)^2\leq \left(1+C_\delta c_{qu}\right)^2 =: N
$$
Therefore,
$$
\frac{\lambda_1}{\lambda_2} \leq \frac{N\delta^2}{(\delta-2h)^2}
 = \frac{NC_\delta^2}{(C_\delta-2)^2} =:C_0
$$

Let $D_1 = [p_1-p_0,p_2-p_0,...,p_n-p_0]^T$, then $M_1 = D_1^TD_1/n$. We consider Cholesky decomposition of $M_1$, one can find upper triangular matrix $U$ with positive diagonal entries such that $M_1 = U^TU$. Construct a map $g(x) = (U^T)^{-1}(x-p_0)+p_0$, which is a quasi-conformal map. Let $M_0$ be the covariance matrix of $g(S)$, we have
$$
M_0 = \frac{(U^T)^{-1}D_1^TD_1U^{-1}}{n} = (U^T)^{-1}M_1U^{-1} = I
$$
Let $\phi = f\circ g^{-1} = (\phi_1, \phi_2)$, and $a_i = g(p_i)$, $a_0 = g(p_0) = \sum_{i=1}^n a_i/n$. Then $\|a_i-a_0\|\leq \sqrt{Trace(nM_0)}\leq \sqrt{2N}$. Since for each point $a_i$, $\phi_1(a_i)-\phi_1(a_0) = \nabla \phi_1(a_0)^T (a_i-a_0) + (a_i-a_0)^T H_i (a_i-a_0)$, for some matrix $H_i=\nabla^2\phi_1(\xi_i)$, then
\begin{align*}
&\left|(M_2)_{11} - \nabla \phi_1(a_0)^T \nabla \phi_1(a_0)\right|\\
 = \ &\left|\frac{1}{n}\sum_{i=1}^n \left(\frac{1}{n} \sum_{j=1}^n (\phi_1(a_i)-\phi_1(a_j))\right)^2
 - \nabla \phi_1(a_0)^T \nabla \phi_1(a_0)\right|\\
 = \ &\left|\frac{1}{n}\sum_{i=1}^n \biggl( (a_i-a_0)^TH_i(a_i-a_0) - \frac{1}{n} \sum_{j=1}^n( (a_j-a_0)^TH_j(a_j-a_0))\biggr.\right.\\
  & +\biggl.\nabla \phi_1(a_0)^T(a_i-a_0)\biggr)^2
\left.-\frac{1}{n}\nabla \phi_1(a_0)^T\left( \sum_{i=1}^n (a_i-a_0)(a_i-a_0)^T\right) \nabla \phi_1(a_0) \right|\\
 \leq\  & \frac{2}{n}\sum_{i=1}^n\left(\left|(a_i-a_0)^TH_i(a_i-a_0)\right| + \frac{1}{n}\sum_{j=1}^n \left|(a_j-a_0)^TH_j(a_j-a_0)\right|\right)e_i
\end{align*}
where
\begin{align*}
e_i = &\left|\nabla \phi_1(a_0)^T(a_i-a_0)\right| + \left|(a_i-a_0)^TH_i(a_i-a_0)\right| \\
&+ \left|\frac{1}{n} \sum_{j=1}^n( (a_j-a_0)^TH_j(a_j-a_0))\right|\\
\leq \ & \sqrt{2}\|\phi_1\|_{C^1(g(\mathcal{D}))}\sqrt{2N} + 4\|\phi_1\|_{C^2(g(\mathcal{D}))}2N
\end{align*}
According to the range of $\lambda_1$ and $\lambda_2$, we have
\begin{align*}
&\|\phi_1\|_{C^1(g(\mathcal{D}))}\leq \|f\|_{C^1(\mathcal{D})}\cdot\|U\|_\infty
\leq \|f\|_{C^1(\mathcal{D})}\sqrt{2Trace(M_1)}
\leq 2\|f\|_{C^1(\mathcal{D})}\delta\\
&\|\phi_1\|_{C^2(g(\mathcal{D}))}
\leq \|f\|_{C^2(\mathcal{D})}\cdot \|U\|_\infty^2
\leq 4\|f\|_{C^2(\mathcal{D})}\delta^2
\end{align*}
Therefore,
\begin{align*}
&\left|(M_2)_{11} - \nabla \phi_1(a_0)^T \nabla \phi_1(a_0)\right|\\
\leq \ & \frac{4}{n}\left(\sum_{i=1}^n \left|(a_i-a_0)^TH_i(a_i-a_0)\right|\right)\left(4\sqrt{N}\|f\|_{C^1(\mathcal{D})}\delta + 32N\|f\|_{C^2(\mathcal{D})}\delta^2\right)\\
\leq \ & 32\|f\|_{C^2(\mathcal{D})}\delta^2\left(4\sqrt{N}\|f\|_{C^1(\mathcal{D})}\delta + 32N\|f\|_{C^2(\mathcal{D})}\delta^2\right)2N\\
\leq \ & C_1h^3
\end{align*}
where $C_1 = 256N\|f\|_{C^2(\mathcal{D})}C_\delta^3\left(\sqrt{N}\|f\|_{C^1(\mathcal{D})} + 8N\|f\|_{C^2(\mathcal{D})}C_\delta h_0\right)$.

With similar argument, one can also prove $\left|(M_2)_{12} - \nabla \phi_1(a_0)^T \nabla \phi_2(a_0)\right| \leq C_1h^3$, and $\left|(M_2)_{22} - \nabla \phi_2(a_0)^T \nabla \phi_2(a_0)\right|\leq C_1h^3$.

Let $M_2 = \nabla \phi(a_0) \nabla \phi(a_0)^T + E$, where $|E_{ij}|\leq C_1h^3$. By solving the eigenvalues of $M_2$, we get
\begin{align*}
&\lambda_3 = \frac{trace(M_2)+\sqrt{trace(M_2)^2 - 4\det{M_2}}}{2}\\
&\lambda_4 = \frac{trace(M_2)-\sqrt{trace(M_2)^2 - 4\det{M_2}}}{2}
\end{align*}
Let $\Lambda = trace(M_2)$, $\Lambda_1 = trace(\nabla \phi(a_0) \nabla \phi(a_0)^T)$, $\Lambda_2 = trace(E)$, and $D=\det M_2$, $D_1=\det(\nabla \phi(a_0) \nabla \phi(a_0)^T)$, $D_2=\det E$. Assume $\nabla \phi(a_0) \nabla \phi(a_0)^T$ has two eigenvalues $\tilde{\lambda}_3$, $\tilde{\lambda}_4$, and $\tilde{\lambda}_3 \geq \tilde{\lambda}_4$.

By the matrix perturbation theory, 
$|\lambda_3-\tilde{\lambda}_3| = w_0^T E w_0 + O(\|E\|_2^2)\leq 2C_1h^3
$
where $w_0$ is the eigenvector of $M_2$ with respect to $\lambda_3$. Similarly, we have $|\lambda_4-\tilde{\lambda}_4|\leq 2C_1h^3$. Then consider the upper and lower bound of $\tilde{\lambda}_3$ and $\tilde{\lambda}_4$.
Notice that $D_1 = \det(\nabla \phi(a_0))^2\geq \det(\nabla f)^2\lambda_2^2$, and $\lambda_2\geq (\delta-2h)^2/n \geq C_2h^2$, where $C_2 = (C_\delta-2)^2/N$.
\begin{align*}
&|\tilde{\lambda}_3| \leq \Lambda_1 \leq 16\|f\|_{C^1(\mathcal{D})}^2\delta^2 = 16C_\delta^2\|f\|_{C^1(\mathcal{D})}^2h^2=:C_3h^2\\
&|\tilde{\lambda}_4| \leq |\tilde{\lambda}_3| \leq C_3h^2\\
&|\tilde{\lambda}_4| = \frac{2D_1}{\Lambda_1 + \sqrt{\Lambda_1^2-4D_1}} \geq \frac{D_1}{\Lambda_1}\geq \frac{L_1^2\lambda_2^2}{C_3h^2}
\geq \frac{L_1^2C_2^2h^2}{C_3}=:C_4h^2
\end{align*}
$L_1 = \min_{x\in \mathcal{D}}\det{\nabla f}$, which is positive since $\mathcal{D}$ is compact and $f$ is orientation-preserving.

Let $\tau$ be the Beltrami coefficient of $f\circ g^{-1}$, then
$$
\frac{\tilde{\lambda}_3 }{ \tilde{\lambda}_4} = \frac{(1+|\tau(a_0)|)^2}{(1-|\tau(a_0)|)^2}
$$
And
\begin{align*}
&\left|\frac{\lambda_3}{\lambda_4} - \left(\frac{1+|\tau(a_0)|}{1-|\tau(a_0)|}\right)^2\right|
\leq \frac{|\tilde{\lambda}_4| |\lambda_3-\tilde{\lambda}_3| + |\tilde{\lambda}_3| |\tilde{\lambda}_4-\lambda_4|}{|\tilde{\lambda}_4\lambda_4|}\\
\leq &\
\frac{4 C_1C_3h}{C_4(C_4 - 2C_1h_0)}=:C_5h
\end{align*}
Let $\sigma$ be the Beltrami coefficient of $g$, According to composition formula,
$$
\tau(a_0) = \frac{\mu(p_0)-\sigma(p_0)}{1-\overline{\sigma(p_0)}\mu(p_0)}\cdot \frac{g_z(p_0)}{\overline{g_z(p_0)}}
$$
In order to calculate $\tau(a_0)$, we calculate $\sigma(p_0)$ first.
By quasi-conformal theory, we have the following argument.
For arbitrary point $x$,
$$
\sigma(x) = \frac{\sqrt{\lambda_1}-\sqrt{\lambda_2}}{\sqrt{\lambda_1}+\sqrt{\lambda_2}}(v_{21}+iv_{22})^2
$$
where $v_2 = [v_{21}, v_{22}]^T$.
Then,
$
|\sigma(x)| \leq (\sqrt{C_0}-1)/(\sqrt{C_0}+1)
$.

Hence,
$$
T = \frac{\tilde{\mu}(p)-\sigma(p)}{1-\overline{\sigma(p)}\tilde{\mu}(p)}
$$
Since $\|\mu\|_\infty<1$, one can choose $h$ small enough such that $|\tilde{\mu}(x)|\leq (1+\|\mu\|_\infty)/2$ for all $x$. Since $\|p_0-p\|\leq \delta$, we have $|\mu(p_0)-\mu(p)|\leq 2\|\mu\|_{C^1(\mathcal{D})} \delta$. Let $C_6$ be the constant such that $|\mu(p) - \tilde{\mu}(p)|\leq C_6h^2$, then
\begin{align*}
\left| |\tau(a_0)| - |T|\right| &\leq
\left|\frac{(\mu(p_0) - \tilde{\mu}(p))(1-|\sigma(p)|^2)}{(1-\overline{\sigma(p)}\mu(p_0))(1-\overline{\sigma(p)}\tilde{\mu}(p))}\right|\\
&\leq \frac{2|\mu(p_0) - \tilde{\mu}(p)|}{(1-\|\mu\|_\infty)^2} \leq \frac{2(C_6h^2 + 2\|\mu\|_{C^1(\mathcal{D})} \delta )}{(1-\|\mu\|_\infty)^2} \leq C_7h
\end{align*}
where $C_7 = 2(C_6h_0 + 2C_\delta\|\mu\|_{C^1(\mathcal{D})})(1-\|\mu\|_\infty)^{-2}$.
Hence,
\begin{align*}
&\left|\left(\frac{1+|\tau(a_0)|}{1-|\tau(a_0)|}\right)^2 -
\left(\frac{1+|T|}{1-|T|}\right)^2\right|\\
\leq &\left| \frac{4(1-|\tau(a_0)|\cdot|T|)\cdot (|\tau(a_0)|-|T|)}{(1-|\tau(a_0)|)^2(1-|T|)^2}\right|\\
\leq &\frac{4C_7h}{(1-|\tau(a_0)|)^2(1-|T|)^2}
\end{align*}
And
\begin{align*}
&1-|\tau(a_0)|
\geq \frac{1}{2}\left(1-\left|\frac{\mu(p_0)-\sigma(p_0)}{1-\overline{\sigma(p_0)}\mu(p_0)}\right|^2\right)\\
=&\frac{(1-|\mu(p_0)|^2)(1-|\sigma(p_0)|^2)}{2|1-\overline{\sigma(p_0)}\mu(p_0)|^2}
\geq  \frac{2(1-\|\mu\|_\infty^2)\sqrt{C_0}}{C_0+1+2\sqrt{C_0}}=:C_8
\end{align*}
One can find $h$ small enough such that $1-|T|\geq C_8/2$.
Then,
$$
\left|\left(\frac{1+|\tau(a_0)|}{1-|\tau(a_0)|}\right)^2 -
\left(\frac{1+|T|}{1-|T|}\right)^2\right|
\leq \frac{8C_7h}{C_8^2} =: C_9h
$$
Therefore,
$$
\left|\frac{\lambda_3}{\lambda_4} - \left(\frac{1+|T|}{1-|T|}\right)^2\right| \leq (C_5 + C_9)h
$$
\qed
\end{proof}

\subsection{Proof of Proposition \ref{Prop:Pca}(b)}

\begin{proof}
Here, we use the same notation as in the above proof. By assumption, when $h$ is small, $\tau(a_0) \neq 0$,
which implies $\nabla \phi(a_0) \nabla \phi(a_0)^T\neq CI$.
Consider the linear transformation $L(x) = \nabla \phi(a_0) (x-a_0) + a_0$.
Then by linear algebra, $L$ maps the unit disk $B(a_0)$ to an ellipse whose major axis direction is parallel to the unit eigenvector $w_1$ of $\nabla \phi(a_0) \nabla \phi(a_0)^T$ with respect to greater eigenvalue. Without loss of generality, we can choose $w_0\cdot w_1\geq 0$.
By quasi-conformal theory, the direction of major axis is parallel to $\nabla \phi(a_0) u'_1$ where $u'_1 = [\cos(\gamma), \sin(\gamma)]^T$ and $\gamma = angle(\tau(a_0))/2$. Therefore, $w_1$ is parallel to $\nabla \phi(a_0) u'_1$.

Here, we define
$$
w_2 = \frac{\nabla \phi(a_0) u'_1}{\|\nabla \phi(a_0) u'_1\|}
$$
By the composition formula for Beltrami coefficient,
$$
\tau(a_0) = \frac{\mu(p_0)-\sigma(p_0)}{1-\overline{\sigma(p_0)}\mu(p_0)}\cdot \frac{g_z(p_0)}{\overline{g_z(p_0)}}
$$
Let $\theta' = angle\left((\mu(p_0)-\sigma(p_0))(1-\overline{\mu(p_0)}\sigma(p_0))\right)/2$, then from composition formula, $\gamma = \theta' + angle(g_z(p_0))$. Denote $U = \bigl[\begin{smallmatrix}
a&b \\ 0&c
\end{smallmatrix} \bigr]$, and $u'_2 = [\cos(\theta'), \sin(\theta')]^T$, then $g_z(x_0) = (a+c-bi)/(ac)$, and $u'_1 = A u'_2$, where
$$
A = \frac{1}{(a+c)^2+b^2}\left[ \begin{matrix}
a+c & b\\
-b & a+c
\end{matrix}\right]
$$
Hence
\begin{align*}
&\nabla \phi(a_0) u'_1 = \nabla f(p_0) U^T A u'_2\\
= \,&\frac{1}{(a+c)^2+b^2} \nabla f(p_0)(M_1 + \sqrt{\lambda_1\lambda_2}I) u'_2
\end{align*}
Therefore,
$$
w_2 =  \frac{\nabla f(p_0)(M_1 + \sqrt{\lambda_1\lambda_2}I) u'_2}{\|\nabla f(p_0)(M_1 + \sqrt{\lambda_1\lambda_2}I) u'_2\|}
$$

By assumption, $w_0\cdot w_1\geq 0$, and $w_0$, $w_1$ are unit eigenvectors related to greater eigenvalues of $M_2$ and $\nabla \phi(a_0) \nabla \phi(a_0)^T$ respectively.
From the proof in above subsection, $M_2 = \nabla \phi(a_0) \nabla \phi(a_0)^T+E$ where $|E_{ij}| \leq C_1h^3$, and $E$ is furthermore symmetric positive definite. By perturbation analysis of eigenvector, let $w_1'$ be the second unit eigenvector of $\nabla \phi(a_0) \nabla \phi(a_0)^T$, then
$$
|w_0-w_1| \leq 2\left|\frac{w_1^TEw_1'}{\eta_1 - \eta_2} w_1'\right| \leq \frac{4C_1h^3}{|\eta_1-\eta_2|}
$$
Here $\eta_1$ and $\eta_2$ are eigenvalues of $\nabla \phi(a_0) \nabla \phi(a_0)^T = \nabla f(p_0)M_1 \nabla f(p_0)^T$. Hence $\eta_1\eta_2 = \lambda_1\lambda_2\det(\nabla f(p_0))^2$, and $\eta_1/\eta_2 = (1+|\tau(a_0)|)^2/(1-|\tau(a_0)|)^2$. Then,
\begin{align*}
|\eta_1 - \eta_2| &=
\det(\nabla f(p_0))\sqrt{\lambda_1\lambda_2}\left(\frac{1+|\tau(a_0)|}{1-|\tau(a_0)|}-\frac{1-|\tau(a_0)|}{1+|\tau(a_0)|}\right)\\
&\geq 2L_1\frac{(\delta-2h)^2}{n}|\mu(p_0)-\sigma(p_0)| \geq L_1C_2h^2|\zeta|
\end{align*}
when $h$ is small enough, where $C_2 = (C_\delta-2)^2/N$, and $L_1$ be the lower bound of $\det(\nabla f(x))$, which is positive. Hence $|w_0-w_1|\leq 4C_1h/(L_1C_2|\zeta|)$.

From the proof in above subsection, we can obtain the following result.
$$
\left|\frac{\mu(p_0)-\sigma(p_0)}{1-\overline{\sigma(p_0)}\mu(p_0)} - \frac{\tilde{\mu}(p)-\sigma(p)}{1-\overline{\sigma(p)}\tilde{\mu}(p)}\right| \leq C_3 h
$$
for some constant $C_3$ depending on $f$.
Therefore, when $h$ is small enough,
\begin{align*}
&|u_0-u'_2|\leq |\theta - \theta'| \\
\leq&\ \frac{\pi}{4} \left|\frac{(\mu(p_0)-\sigma(p_0))(1-\overline{\mu(p_0)}\sigma(p_0))}{\|(\mu(p_0)-\sigma(p_0))(1-\overline{\mu(p_0)}\sigma(p_0))\|} - \frac{(\tilde{\mu}(p)-\sigma(p))(1-\overline{\tilde{\mu}(p)}\sigma(p))}{\|(\tilde{\mu}(p)-\sigma(p))(1-\overline{\tilde{\mu}(p)}\sigma(p))\|}\right|\\
\leq&\ \frac{\pi}{4}\cdot 2C_3h \cdot \left|\frac{1-\overline{\sigma(p_0)}\mu(p_0)}{\mu(p_0)-\sigma(p_0)}\right|\\
\leq&\ \frac{C_3 \pi h}{|\mu(p_0) - \sigma(p_0)|} \leq \frac{C_3\pi h}{2|\zeta|}
\end{align*}
Then,
\begin{align*}
&\ \left|\nabla f(p_0)(M_1+\sqrt{\lambda_1\lambda_2}I)u'_2 - \nabla f(p)(M_1+\sqrt{\lambda_1\lambda_2}I )u_0\right|\\
\leq &\ \left|(\nabla f(p_0)-\nabla f(p))(M_1+\sqrt{\lambda_1\lambda_2}I)u'_2\right|+\left| \nabla f(p)(M_1+\sqrt{\lambda_1\lambda_2}I )(u'_2-u_0)\right|\\
\leq & \ 2\sqrt{2}\|f\|_{C^2(\mathcal{D})}\|x_0-p\|(\lambda_1+\sqrt{\lambda_1\lambda_2})
+ 2\|f\|_{C^1(\mathcal{D})}(\lambda_1+\sqrt{\lambda_1\lambda_2}) \frac{C_3\pi h}{2|\zeta|}\\
\leq &\ 2\delta^2\left(2\sqrt{2}C_\delta\|f\|_{C^2(\mathcal{D})}
+ \frac{C_3\pi \|f\|_{C^1(\mathcal{D})}}{|\zeta|}\right) h =:C_4h^3 + \frac{C_5h^3}{|\zeta|}
\end{align*}
Let $\lambda$ be the smaller eigenvalue of $\nabla f(p) \nabla f(p)^T$, then we get the following result. Notice that $\lambda_2 \geq C_2h^2$, and $\lambda = \det(\nabla f(p))^2 (1-|\mu(p)|)^2(1+|\mu(p)|)^{-2}$.
\begin{align*}
&\ \left| w_2 - \frac{\nabla f(p)(M_1 + \sqrt{\lambda_1\lambda_2}I) u_0}{\|\nabla f(p)(M_1 + \sqrt{\lambda_1\lambda_2}I) u_0\|}\right|\\
\leq &\ \frac{2C_4h^3+2C_5h^3|\zeta|^{-1}}{\|\nabla f(p)(M_1 + \sqrt{\lambda_1\lambda_2}I) u_0\|}\\
\leq &\ \frac{2C_4h^3+2C_5h^3|\zeta|^{-1}}{\sqrt{\lambda}(\lambda_2 + \sqrt{\lambda_1\lambda_2})}\\
\leq &\ \frac{(C_4h^3+C_5h^3|\zeta|^{-1})(1+|\mu(p)|^2)}{C_2h^2(1-|\mu(p)|^2)\det(\nabla f(p))}\\
\leq &\ \frac{(C_4h+C_5h|\zeta|^{-1})(1+\|\mu\|_\infty^2)}{C_2(1-\|\mu\|_\infty^2)L_1}\\
=&: C_6\frac{h}{|\zeta|}
\end{align*}

By assumption, $w_0\cdot (\nabla f(p) (M_1 + \sqrt{\lambda_1\lambda_2}I) u_0) \geq 0$. And from the above proof, we have the following argument. When $h$ and $h/|\zeta|$ are small enough, $w_2$ is close to direction vector of $\nabla f(p) (M_1 + \sqrt{\lambda_1\lambda_2}I) u_0$, $w_0$ is close to $w_1$, and $w_1$ is parallel to $w_2$.
Hence $w_1=w_2$. Therefore,
\begin{align*}
&\ \left| w_0 - \frac{\nabla f(p)(M_1 + \sqrt{\lambda_1\lambda_2}I) u_0}{\|\nabla f(p)(M_1 + \sqrt{\lambda_1\lambda_2}I) u_0\|}\right|\\
\leq &\ \frac{4C_1h}{L_1C_2|\zeta|} + C_6\frac{h}{|\zeta|}
=: \ C_7h + C_6\frac{h}{|\zeta|}
\end{align*}
where $C_6$ and $C_7$ only depends on $f$.
\qed
\end{proof}

\bibliographystyle{abbrv}
\bibliography{biborder}{}

\end{document}